\documentclass[11pt]{article}

\usepackage{amsmath,times,amssymb}

\usepackage{hyperref}

\hypersetup{colorlinks=false}
\usepackage{amsthm}

\usepackage{color}

\usepackage{thmtools}
\usepackage{cleveref}
\usepackage{stmaryrd}

\newtheorem{theorem}{Theorem} [section]
\newtheorem{lemma}[theorem]{Lemma}
\newtheorem{definition}[theorem]{Definition}
\newtheorem{corollary}[theorem]{Corollary}

\newtheorem{proposition}[theorem]{Proposition}

\newtheorem{observation}[theorem]{Observation}

\usepackage{graphicx}
\usepackage{url} 
\usepackage[left=1.0in,top=1.0in,right=1.0in,bottom=1.0in,nohead,textheight=10in,footskip=0.3in]{geometry}

\usepackage[noend]{algpseudocode} 
\usepackage[nothing]{algorithm}  

\makeatletter
\renewcommand\footnotesize{%
   \@setfontsize\footnotesize\@ixpt{8}%
   \abovedisplayskip 8\p@ \@plus2\p@ \@minus4\p@
   \abovedisplayshortskip \z@ \@plus\p@
   \belowdisplayshortskip 4\p@ \@plus2\p@ \@minus2\p@
   \def\@listi{\leftmargin\leftmargini
               \topsep 4\p@ \@plus2\p@ \@minus2\p@
               \parsep 2\p@ \@plus\p@ \@minus\p@
               \itemsep \parsep}%
   \belowdisplayskip \abovedisplayskip
}
\makeatother

\newcommand{\bE}{\ensuremath{\mathbb{E}}}

\newcommand{\ignore}[1]{}
\newcommand\oa{\overline{\mathcal A}}
\newcommand\obr{\overline{\mathbf R}}

\newcommand{\eps}{\epsilon}

\DeclareMathOperator{\poly}{poly}
\DeclareMathOperator{\sink}{sink}
\DeclareMathOperator{\depth}{depth}

\usepackage{calc}
\usepackage{cleveref}
\begin{document}

\thispagestyle{empty}

\title{A new notion of commutativity  for the algorithmic \\ Lov\'{a}sz Local Lemma}

\author{
David G. Harris \\
University of Maryland, College Park\\
{\small davidgharris29@gmail.com}
\and
Fotis Iliopoulos\thanks{This material is based upon work directly supported by the IAS Fund for Math and indirectly supported by the National Science Foundation Grant No. CCF-1900460. Any opinions, findings and conclusions or recommendations expressed in this material are those of the author(s) and do not necessarily reflect the views of the National Science Foundation. This work is also supported by the National Science Foundation Grant No. CCF-1815328.} \\ 
Institute for Advanced Study \\
{\small fotios@ias.edu}
\and
Vladimir Kolmogorov\thanks{Supported by the European Research Council under
the European Unions Seventh Framework Programme (FP7/2007-2013)/ERC
grant agreement no 616160} \\
Institute of Science and Technology Austria \\ {\small vnk@ist.ac.at}
}

\maketitle 

\begin{abstract}
The Lov\'{a}sz Local Lemma (LLL) is a powerful tool in probabilistic combinatorics which can be used to establish the \emph{existence} of objects that satisfy certain properties. The breakthrough paper of Moser \& Tardos and follow-up works revealed that the LLL has intimate connections with a class of stochastic local search algorithms for finding such desirable objects. In particular, it can be seen as a sufficient condition for this type  of algorithm to converge fast.  

Besides conditions for convergence, many other natural questions can be asked about algorithms; for instance, ``are they parallelizable?", ``how many solutions can they output?", ``what is the expected `weight' of a solution?". These questions and more have been answered for a class of LLL-inspired algorithms called commutative. In this paper we introduce a new, very natural and more general  notion  of commutativity (essentially matrix commutativity) which allows us to show a number of new refined properties of LLL-inspired local search algorithms with significantly simpler proofs.
\end{abstract}


\section{Introduction}
The Lov\'{a}sz Local Lemma  (LLL) is a powerful tool in probabilistic combinatorics \cite{LLL}. At a high level, it states that given a collection of bad events in a probability space $\mu$, where each bad-event is  not too likely and is independent of most other bad events, there is a strictly positive probability of avoiding all of them. In particular, a configuration avoiding all such bad-events exists. 

In its simplest, ``symmetric'' form, the LLL requires that each bad-event has probability at most $p$ and is dependent with at most $d$ others, where $ e p d \leq 1$.  

 For example, consider a CNF formula where each clause has $k$ literals and shares variables with at most $L$ other clauses. For each clause $c$ we can define the bad event $B_c$ that $c$ is violated in a chosen assignment of the variables. For a uniformly random variable assignment, each bad-event has probability $p = 2^{-k}$ and affects at most $d \leq L$ others. So when $L \leq \frac{2^k}{e}$, the formula is satisfiable; crucially, this criterion does not depend on the total number of variables.

A generalization known as the Lopsided LLL (LLLL) allows bad-events to be \emph{positively} correlated (in a certain technical sense) instead of independent. For example, consider an $n \times n$ array of colors, where each color appears at most $\Delta$ times in total. We wish to find a \emph{latin transversal} of $C$, that is, a permutation $\pi$ over $\{1, \dots, n \}$ such that all colors $C(i, \pi i): i = 1, \dots, n$ are distinct. To apply the LLLL, we take our probability space to be the uniform distribution on permutations $\pi$. For each pair of cells $(x_1, y_1), (x_2, y_2)$ of the same color, there is a corresponding bad-event $\pi x_1 = y_1 \wedge \pi x_2 = y_2$, which has probability $p = \frac{1}{n(n-1)}$. Erd\H{o}s \& Spencer \cite{LopsTrav} showed that two bad-events here are negatively-correlated only if they overlap on a row or column. So each bad-event is dependent with at most $d = 4 (\Delta-1) n$ others.  Thus, for $\Delta \leq \frac{n}{4 e}$, the LLLL criterion holds and a transversal exists. A stronger form of the LLLL (the cluster-expansion criterion \cite{bissacot}) tightens this to $\Delta \leq \frac{27n}{256}$, which is the strongest bound currently known.

Although the LLLL applies to general probability spaces, most constructions in combinatorics use a simpler setting we refer to as the \emph{variable LLL}. Here, the probability space $\mu$ is a cartesian product with $n$ independent variables $X = (X_1, \dots, X_n)$, and each bad-event is determined by a subset of the variables. Two bad-events are dependent if they share a common variable. This covers, for instance, the CNF formula example described above. In a seminal work,  Moser \& Tardos~\cite{MT} presented a simple local search algorithm to make the variable-version LLL constructive. This algorithm can be described as follows:
\begin{algorithm}[H]
\centering
\label{gen-alg}
\begin{algorithmic}[1]
  \State Draw the state $X$ from distribution $\mu$
\While{some bad-event is true on $X$}
\State Select, arbitrarily, a bad-event $B$ true on $X$
\State Resample, according to distribution $\mu$, all variables $X_i$ affecting $B$
\EndWhile
\end{algorithmic}
\caption{The Moser-Tardos (MT) resampling algorithm}
\label{mt-alg}
\end{algorithm}

Moser \& Tardos showed that if the symmetric LLL criterion (or more general asymmetric LLLL criterion) is satisfied, then this algorithm quickly converges. Following this work, a large effort has been devoted to making different variants  of the LLLL constructive. This research has taken many directions, including further analysis of \Cref{mt-alg} and its connection to different LLL criteria~\cite{determ,szege_meet,PegdenIndepen}.

One line of research has been to use variants of the MT algorithm for general probability spaces beyond the variable LLL. We can summarize this in the following framework. There is a discrete state space $\Omega$, with a collection $\mathcal{F}$ of subsets  of $\Omega$ which we call {\it flaws}. There is also a randomized procedure called the \emph{resampling oracle} $\mathbf R_f$ for each flaw $f$; it takes some random action to attempt to ``fix'' that flaw, resulting in a new state $\sigma' \leftarrow \mathbf R_f(\sigma)$. (It is possible that this new state $\sigma'$ does not actually fix $f$.) With these ingredients, we define the general local Search Algorithm as follows:

\begin{algorithm}[H]
\centering
\begin{algorithmic}[1]
  \State Draw the state $\sigma$ from some distribution $\mu$
\While{some flaw holds on $\sigma$}
\State Select a flaw $f$ of $\sigma$, according to some rule $\mathbf S$.
\State Update $\sigma \leftarrow  \mathbf R_f(\sigma)$.
\EndWhile
\end{algorithmic}
\caption{The Search Algorithm}
\label{ls-alg}
\end{algorithm}

Again consider the latin transversal construction, which was one of the main original motivations behind the  Search Algorithm \cite{PermHarris}. For this problem, each pair of cells with the same color now corresponds to a flaw. The resampling oracle for a flaw applies a random ``swapping'' operation  to the permutation. This algorithm generates a latin transversal under the same conditions as the existential LLLL argument.  Further application of the Search Algorithm include matchings and spanning trees of the clique  \cite{AIJACM, AIK,PermHarris,IS,HV} as well as settings not directly connected to the LLL \cite{AIS,PartResmp2,PartResmp1}.
 
The most important question about the behavior of the Search Algorithm is whether it converges to a flawless object. But, there are other important questions to ask; for instance, ``is it parallelizable?", ``how many solutions can it output?", ``what is the expected  `weight' of a solution?". These questions and more have been answered for the MT algorithm in a long series of papers~\cite{determ,distributed,ParallelHarrisHaeupler,Haeupler_jacm,ParallelHarris,EnuHarris,szege_meet,MT}. 
For example,  results of \cite{Haeupler_jacm,NewBoundsHarris,EnuHarris} discuss how to estimate the entropy of the output distribution and how to deal algorithmically with super-polynomially many bad events.

We emphasize the role of the selection rule $\mathbf S$ in the Search Algorithm. This procedure must choose which flaw $f \ni \sigma$ to resample, if there are multiple possibilities; it may depend on the prior states and may be randomized. The original MT algorithm allows nearly complete freedom for this. However, for general resampling oracles, convergence guarantees are only known for a few relatively rigid rules such as selecting the flaw of least index~\cite{HV}.  In \cite{Kolmofocs}, Kolmogorov identified a property called \emph{commutativity} that allows a free choice for $\mathbf S$ in the Search Algorithm (as in the MT algorithm). This seemingly minor detail turns out to play a key role in extending the additional properties of the MT algorithm to the Search Algorithm. For instance, it leads to parallel algorithms \cite{Kolmofocs} and to bounds on the output distribution \cite{LLLWTL}.   

At a high level, the goal of this paper is to provide a more conceptual, algebraic explanation for the commutativity properties of resampling oracles and their role in the Search Algorithm. We do this by introducing a notion of commutativity, essentially matrix commutativity, that is both more general and simpler than the definition in \cite{Kolmofocs}.  Most of our results had already been shown, in slightly weaker forms, in prior works \cite{Kolmofocs,LLLWTL, NewBoundsHarris}.  However, the proofs were computationally heavy and narrowly targeted to certain probability spaces,  with numerous technical side conditions and restrictions. 
  
Before the formal definitions,  let us provide some intuition. For each flaw $f$, consider an $| \Omega | \times |\Omega |$ transition matrix $A_f$. Each row of $A_f$   describes the probability distribution of resampling $f$ at a given state $\sigma$.  We call the resampling oracle \emph{commutative} if the transition matrices commute for any pair of flaws which are ``independent'' in the LLL sense. We show a number of results for such  commutative oracles:
\begin{itemize}

\item We obtain bounds on the distribution of the state at the termination of the Search Algorithm. These bounds are comparable to the \emph{LLL-distribution}, i.e., the distribution induced by conditioning on avoiding all bad events. Similar results, albeit with a number of additional technical restrictions, had been shown in \cite{LLLWTL} for the original definition of commutativity. 

\item We develop a generic parallel version of the Search Algorithm. This extends results of \cite{Kolmofocs, ParallelHarris}, with simpler and more general proofs.

\item In many settings, flaws are formed from smaller ``atomic'' events \cite{ParallelHarris}. We show that, if the atomic events satisfy the generalized commutativity definition, then so do the larger ``composed'' events. This natural property did not seem to hold for the original commutativity definition of \cite{Kolmofocs}.

\item For some probability spaces, specialized distributional  bounds  are available,  beyond the ``generic'' LLL bounds \cite{NewBoundsHarris}.  Our construction captures many of these results in a stronger and more unified way.
\end{itemize}

As a concrete example of our results, we show that, for the latin transversal application, the resulting permutation $\pi$ has nice distributional properties. In particular, we show the following:

\begin{theorem}
\label{perm-thm44}
If each color appears at most $\Delta \leq \frac{27}{256} n$ times in the array, then the Search Algorithm generates a latin transversal $\pi$ such that, for every pair $(x,y)$, it holds that $$
0.53/n \leq \Pr( \pi x = y ) \leq 1.59/n
$$
\end{theorem}

The upper bound improves quantitatively over a similar bound of \cite{NewBoundsHarris}; to the best of our knowledge, no non-trivial lower bound was previously known. Intriguingly, these bounds are not known  for the LLL-distribution itself.  To better situate \Cref{perm-thm44}, note that Alon, Spencer, \& Tetali \cite{alon1995covering} showed that there is a (minuscule) universal constant $\beta > 0$ such that, if each color appears at most $\Delta = \beta n$ times in the array and $n$ is a power of two, then the array can be \emph{partitioned} into $n$ independent transversals. In this case, randomly selecting  a transversal from this list would give  $\Pr( \pi x = y) = 1/n$. \Cref{perm-thm44} can be regarded as a simplified \emph{fractional} analogue, i.e. we fractionally decompose the given array into  transversals, with $\Pr( \pi x = y) = \Theta(1/n)$ for all pairs $x,y$. Furthermore, we achieve this guarantee automatically, merely by running the Search Algorithm.

\subsection{Overview of our approach}

Although it will require significant definitions and technical development to state our results formally, let us provide a high-level summary here. As a starting point, consider the MT algorithm. Moser \& Tardos \cite{MT} used a construction called a \emph{witness tree} for the analysis: for each resampling of a bad-event $B$, they generate a corresponding witness tree which records an ``explanation'' of why $B$ was true at that time. More properly, it provides a history of all the prior resamplings which affected the variables involved in $B$.  

The main technical lemma governing the behavior of the MT algorithm is the ``Witness Tree Lemma,'' which states that the probability of producing a given witness tree is at most the product of the probabilities of the corresponding events. The bounds on algorithm runtime, as well as parallel algorithms and distributional properties, then follow from a union bound over witness trees. 

Versions of this Witness Tree Lemma have been shown for some variants of the MT algorithm \cite{LLLLBeyond, PartResmp1}. Iliopoulos \cite{LLLWTL} further showed that it held for general spaces which satisfy the commutativity property; this, in turn, leads to the nice algorithmic properties such as parallel algorithms.

Our main technical innovation is to generalize the Witness Tree Lemma. Instead of tracking a \emph{scalar} product of probabilities in a witness tree, we instead consider a \emph{matrix product.}  We bound  the probability of a given witness tree (or, more properly, a slight generalization known as a witness directed acyclic graph) in terms of the products of the transition matrices of the corresponding flaws.  At the end, we obtain the scalar form of the Witness Tree Lemma by projecting to a one-dimensional eigenspace; for this, we take advantage of some spectral estimation methods of \cite{AIS}.

\subsection{Outline of the paper} 

In \Cref{Background}, we introduce our new matrix-based definition for commutativity.

In \Cref{WitnessTree-sec}, we define the witness directed acyclic graph following \cite{ParallelHarrisHaeupler}.  We show bounds on the Search Algorithm in terms of certain associated matrix products. We also discuss how these relate to standard ``scalar'' LLL criteria.

In \Cref{distrib-sec}, we derive simple bounds on the state distribution  of the Search Algorithm.

In \Cref{sec:parallel}, we describe a parallel implementation of the Search Algorithm.

In \Cref{comp-sec}, we consider a construction for building resampling oracles out of smaller ``atomic'' events. 

In \Cref{sec:detailedist}, we consider more involved distributional bounds, including applications to latin transversals and clique perfect matchings.

\subsection{Definitions}
Throughout the paper we consider implementations of the Search Algorithm. We list the resampled flaws in order as $f_0, f_1, f_2, \dots$; to avoid ambiguity, we refer to the \emph{state at time $t$} as the state $\sigma$ just \emph{before} resampling flaw $f_t$. Note that the state at time 0 is drawn directly from $\mu$.

For flaw $f$ and states $\sigma \in f, \sigma' \in \Omega$ we write $\sigma \xrightarrow{f} \sigma'$ to denote that the algorithm  resamples $f$ at $\sigma$ and moves to $\sigma'$.  We define $A_f[\sigma, \sigma']$ to be the probability of such transition under resampling oracle $\mathbf R_f$, i.e.
$$
A_f[\sigma, \sigma'] = \Pr( \mathbf R_f( \sigma) = \sigma') 
$$

To allow us to write our state transitions from left-to-right, we also write $e_{\sigma}^{\top} A_f e_{\sigma'} = A_f[\sigma, \sigma']$.  For $\sigma \notin f$, we define $A_f[\sigma, \sigma'] = 0$.   Note that any vector $\theta$ over $\Omega$ satisfies $|| \theta^{\top} A_f ||_1 = \sum_{\sigma \in f} \theta[\sigma] \leq || \theta^{\top} ||_1$. Thus, the matrix $A_f$ is substochastic.

For an arbitrary event $E \subseteq \Omega$, we define $e_E$ to be the indicator vector for $E$, i.e. $e_E[\sigma] = 1$ if $\sigma \in E$ and $e_E[\sigma] = 0$ otherwise. For a state $\sigma \in \Omega$, we write $e_{\sigma}$ as shorthand for $e_{ \{ \sigma \}}$, i.e. the basis vector which has a $1$ in position $\sigma$ and zero elsewhere. With this notation,  $e_{\sigma}^{\top} A_f$ for $\sigma \in f$ is the vector representing the probability distribution obtained by resampling flaw $f$ at state $\sigma$. When $\sigma \notin f$, then $e_{\sigma}^{\top} A_f = \vec 0$.

For vectors $u, v$ we write $u \preceq v$ if $u[i] \leq v[i]$ for all entries $i$. We write $u \propto v$ if there is some scalar value $c$ with $u = c v$. Likewise, for matrices $A, B$ we write $A \propto B$ if $A = c B$ for some scalar value $c$.

\section{Commutativity}
\label{Background} 
 We suppose that we have fixed a symmetric relation $\sim$ on $\mathcal F$, with $f \not \sim f$ for all $f$. We refer to $\sim$ as the \emph{dependence relation}. We define $ \Gamma(f)$ to be the set of flaws $g$ with $f \sim g$, and we also define $\overline \Gamma(f) = \Gamma(f) \cup  \{f \}$. We write $f \simeq g$ if $f \sim g$ or $f = g$.

The major contribution of this paper is to introduce a natural definition for commutativity:

\begin{definition}[Matrix commutativity]\label{commutative_algos}
The resampling oracle is \emph{matrix-commutative} with respect to  dependence relation $\sim$ if  $A_f A_g = A_g A_f$, for every pair of flaws $f,g$ such that $f \nsim g$. 
\end{definition}

For contrast, let us recall the definition of commutativity from \cite{Kolmofocs}. To avoid confusion, we refer to this other notion as \emph{strong commutativity}.

\begin{definition}[Strong commutativity \cite{Kolmofocs}]\label{commutative_algos1}
The resampling oracle is \emph{strongly commutative} with respect to  dependence relation $\sim$ if  for every pair of flaws $f, g $ such that $f \nsim g$, there is an injective mapping from transitions $\sigma_1 \xrightarrow{f} \sigma_2 \xrightarrow{g}  \sigma_3$ to  transitions $\sigma_1 \xrightarrow{g} \sigma_2' \xrightarrow{f} \sigma_3$, so that  $A_f[\sigma_1,\sigma_2] A_g[\sigma_2, \sigma_3] = A_g[\sigma_1,\sigma_2'] A_f[\sigma_2',\sigma_3]$.
\end{definition}

\begin{observation}\label{implies_original} 
  If the resampling oracle is strongly commutative, then it is matrix-commutative.
\end{observation}
\begin{proof}

  Consider $f \not \sim g$. By symmetry, we need to show that $(A_f A_g)[ \sigma, \sigma'] \leq (A_g A_f)[ \sigma, \sigma']$ for any states $\sigma, \sigma'$.  Let $V$ denote the set of states $\sigma''$ with $A_f[\sigma, \sigma''] A_g[ \sigma'', \sigma'] > 0$. By definition, there is an injective function $F: V \rightarrow \Omega$ such that $A_f[\sigma,\sigma''] A_g[\sigma'', \sigma'] = A_g[\sigma,F(\sigma'')] A_f[F(\sigma''),\sigma']$. Therefore, we have
  $$
  (A_f A_g) [ \sigma, \sigma' ] = \sum_{\sigma'' \in V} A_f[\sigma, \sigma''] A_g[\sigma'', \sigma'] = \sum_{\sigma'' \in V} A_g[\sigma,  F(\sigma'')] A_f[ F(\sigma''), \sigma']
  $$

  Since $F$ is injective, each term $A_g[\sigma, \tau] A_f[ \tau, \sigma']$ is counted at most once in this sum with $\tau = F(\sigma'')$. So  $ (A_f A_g) [ \sigma, \sigma' ] \leq \sum_{\tau \in f} A_g[\sigma,\tau] A_f[ \tau, \sigma'] = (A_g A_f) [\sigma, \sigma']$.
  \end{proof}
  
 \begin{observation}
 \label{cause-obs}
 Suppose the resampling oracle is matrix-commutative. If resampling flaw $f$ can cause flaw $g$, then  $f \sim g$.
\end{observation}
\begin{proof}
Consider some state transition $\sigma \xrightarrow{f} \sigma'$ for $\sigma \notin g, \sigma' \in g$. So $e_{\sigma}^{\top} A_f e_g > 0$. Then  $e_{\sigma}^{\top} A_f A_g  \neq \vec 0 = e_{\sigma}^{\top} A_g A_f$. So $A_f A_g \neq A_g A_f$. 
\end{proof}

\textbf{For the remainder of this paper, we use the word ``commutative'' to mean matrix-commutative. We always assume that the resampling oracle $\mathbf R$ is matrix-commutative unless explicitly stated otherwise.}

\bigskip

We say that a set or multiset $I \subseteq \mathcal F$ is \emph{stable} if $f \not \sim g$ for all distinct pairs $f, g \in I$. We define
$$
A_I = \prod_{f \in I} A_f;
$$
here, if $I$ is a multiset, each $f \in I$ is counted  with its multiplicity in $I$. When the resampling oracle is commutative, this is well-defined (without specifying ordering of $I$) since the matrices $A_f$ all commute.

For a stable multiset $I = \{g_1, \dots, g_t \}$, we define $\langle I \rangle = g_1 \cap \dots \cap g_t$. For a flaw $f$, we write $f \sim I$ if there exists any $g \in I$ with $f \sim g$. Similarly, for stable multisets $I, J$, we write $I \sim J$ if there exist $f \in I, g \in J$ with $f \sim g$.

\section{Witness directed acyclic graphs and matrix bounds}\label{WitnessTree-sec}

Following \cite{ParallelHarrisHaeupler}, the  \emph{witness directed acyclic graph} is our key tool to analyze the Search Algorithm.
\begin{definition}[Witness directed acyclic graph]
A  \emph{witness directed acyclic graph} (abbreviated \emph{wdag}) is a directed acyclic graph $H$, where each vertex $v \in H$ has a label $L(v)$ from $\mathcal F$, and such that for all pairs of vertices $v, w \in H$, there is an edge between $v$ and $w$ (in either direction) if and only if $L(v) \simeq L(w)$.  
\end{definition}

  For intuition, every node $v$ in a wdag corresponds to a resampled flaw $f$, with $L(v) = f$. The edges always point \emph{forward in time}: there is an edge from $v$ to $w$ if the flaw corresponding to $v$ was resampled before $w$. For example, consider a scenario with five resampled flaws $f_1, f_2, f_3, f_1, f_5$ in order (where $f_4 = f_1$ has been resampled twice). We could represent this with the following wdag:

\begin{figure}[H]
\vspace{1.9in}
\begin{center}
\hspace{-3.5in}
\includegraphics[trim = 0.5cm 22.0cm 9cm 5cm,scale=1.0,angle = 0]{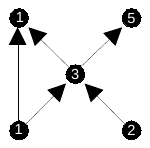}
\vspace{-1.5in}
\caption{\label{fig88} A wdag representing the five indicated flaws.}
\end{center}
\end{figure}
\vspace{-0.2in}

We define $\mathfrak W$ to be the set of wdags. The number of nodes in a wdag $H$ is denoted by $|H|$. Given a flaw $f$ and wdag $H$, we write $f \sim H$ if there exists any node $v \in H$ with $L(v) \sim f$. 

For a wdag $H$ with sink nodes $v_1, \dots, v_k$ labeled $f_1 = L(v_1), \dots, f_k = L(v_k)$, observe that $f_1, \dots, f_k$ are distinct and  that  $\{ f_1, \dots, f_k \}$ is a stable set of flaws; we denote it by $\sink(H)$, and we define $\mathfrak W(I)$ to be the set of wdags with $\sink(H) = I$. We also write $\mathfrak W(f)$ as shorthand for $\mathfrak W( \{f \})$.  We define $\mathfrak S = \bigcup_{f \in \mathcal F} \mathfrak W(f)$ to be the set of \emph{single-sink} wdags.

There is a key connection between wdags and  transition matrices. For any wdag $H$, we  define an associated $|\Omega| \times |\Omega|$ matrix $A_H$ inductively, as follows. If $H = \emptyset$, then $A_H$ is the identity matrix on $\Omega$. Otherwise, we choose an arbitrary source node $v$ of $H$ and set $A_H = A_{L(v)} A_{H - v}$.

\begin{proposition}
  \label{ah-prop1}
The definition of $A_H$ does not depend on the chosen source node $v$. 
\end{proposition}
\begin{proof}
We show it by induction on $|H|$.  When $|H| = 0$ it is vacuous. For the induction step, suppose $H$ has two source nodes $v_1, v_2$. We need to show that we get the same result by recursing on $v_1$ or $v_2$, i.e $A_{L(v_1)} A_{H - v_1} = A_{L(v_2)} A_{H - v_2}$.

  We can apply the induction hypothesis to $H - v_1$ and $H - v_2$, noting that $v_2$ is a source node of $H - v_1$ and $v_1$ is a source node of $H - v_2$. We get $  A_{H - v_1} = A_{L(v_2)} A_{H - v_1 - v_2}, A_{H - v_2} = A_{L(v_1)} A_{H - v_1 - v_2}$.   Thus, in order to show $A_{L(v_1)} A_{H - v_1} = A_{L(v_2)} A_{H - v_2}$, it suffices to show that $A_{L(v_1)} A_{L(v_2)} = A_{L(v_2)} A_{L(v_1)}$.   Since $v_1, v_2$ are both source nodes, we have $L(v_1) \not \sim L(v_2)$. Thus, this follows from commutativity.
\end{proof}

\begin{observation}
\label{ah-prop2}
If $f \not \sim H$ for a flaw $f$ and wdag $H$, then $A_f A_H = A_H A_f$.
\end{observation}
\begin{proof}
We can write $A_H = A_{L(v_1)} \cdots A_{L(v_t)}$ where $v_1, \dots, v_t$ are the nodes of $H$. By hypothesis, we have $L(v_i) \not \sim f$ for all $i$ and the matrices $A_{L(v_i)}$ all commute with $A_f$. So $A_f A_H = A_f  A_{L(v_1)} \cdots A_{L(v_t)} = A_{L(v_1)} \cdots A_{L(v_t)} A_f  = A_H A_f.$
\end{proof}

 One important way of generating wdags is the following: we say that we   \emph{prepend} a given flaw $f$ to a given wdag $H$ when we add a node labeled $f$ with an edge to each other node $v$ with $L(v) \simeq f$.  This operation always results in a new wdag $H'$, which has a source node labeled $f$. Intuitively, this means that $f$ happens before all the events corresponding to the nodes in $H$.

\begin{observation}
If wdag $H'$ is obtained by prepending flaw $f$ to wdag $H$, then $A_{H'} = A_f A_H$.
\end{observation}

\subsection{Wdags and their role in the Search Algorithm}
\label{wdag-construct}
Consider an execution of the Search Algorithm. For each time $t$ with resampled flaw $f_t$, we define a corresponding wdag $G_{t}$, as follows: initially, we set $G_t$ to consist of a singleton node labeled $f_t$. Then, for each time $s = t-1, \ldots, 0$ with resampled flaw $f_s$, there are two cases:
\begin{itemize}
\item If $f_s \sim G_t$, or if $G_t$ has a source node labeled $f_s$,  then prepend $f_s$ to $G_t$.
\item Otherwise, do not modify $G_t$.
 \end{itemize}
 
Note that $\sink(G_t) = \{f_t \}$.  We say a wdag $H$ \emph{appears} if $H$ is isomorphic to some $G_t$; with a slight abuse of notation, we write this simply as $G_t = H$.  
 
  This construction is very similar to the Witness Tree in Moser \& Tardos \cite{MT}.  One of the main ingredients in their proof is the observation that a witness tree $G$ appears with probability at most $\prod_{v \in G } \mu(L(v) ) $, i.e.,  the product of probabilities of the flaws that label its vertices.  (Recall  that $\mu$ denotes the initial state distribution.) Their proof depends on properties of the variable setting and does not extend to other probability spaces.  Our key message is that the new commutativity definition allows us to analogously bound the probability of appearance of a given wdag by the product of transition matrices. Specifically, we show the following. 
  
\begin{lemma}
  \label{witness-tree-lemma_simple}
For any wdag $H$, the probability that $H$ appears is at most $\mu^{\top} A_H \vec 1$.
\end{lemma}
\begin{proof}
We first show that if the Search Algorithm runs for at most $t_{\max}$ steps starting with state $\sigma$, where $t_{\max}$ is an arbitrary integer, then $H$ appears with probability at most $e_{\sigma}^{\top} A_H \vec 1.$ We prove this claim by induction on $t_{\max}$.  The claim holds vacuously for $t_{\max} = 0$, or if $\sigma$ is flawless. Also, if $H$ is a singleton node labeled $f$ and $\sigma \in f$, then $e_{\sigma}^{\top} A_H \vec 1 = 1$ and the claim is vacuous.  

So suppose that $t_{\max} > 0$ and $\sigma$ has a flaw, and that $H$ is not a singleton node with label $f \ni \sigma$.  Then $G_0 \neq H$. So, if $H$ appears,  we have $G_t = H$ for  $t  \in \{1, \dots,  t_{\max} -1 \}$. Now suppose we condition on $\mathbf S$ selecting a flaw $f$ to resample in $\sigma$.  We can view the evolution of the Search Algorithm $\mathbf A$ as a two-part process: first resample $f$, reaching a state $\sigma'$, wherein each potential choice of $\sigma'$ is chosen with probability $A_f[\sigma, \sigma']$. Then execute a new search algorithm $\mathbf A'$ starting at state $\sigma'$, wherein the flaw selection rule $\mathbf S'$ on history $(\sigma', \sigma_2,  \dots, \sigma_t)$ is the same as the choice of $\mathbf S$ on history $(\sigma, \sigma', \sigma_2, \dots, \sigma_t)$.   

 Let $G'_{t-1}$ be the corresponding wdag for $\mathbf A'$. We obtain $G_t$ from $G'_{t-1}$ either by prepending $f$, or by setting $G_t = G'_{t-1}$ when $f \not \sim G'_{t-1}$ and $G'_{t-1}$ has no source node labeled $f$. Consequently, $H$ must satisfy one of the following two mutually exclusive conditions: (i) $H$ has a unique source node $v$ labeled $f$ and $G'_{t-1} = H - v$; or  (ii) $f \not \sim H$ and  $G'_{t-1} = H$.

In case (i), if $H$ appears for the original search algorithm $\mathbf A$ within $t_{\max}$ timesteps, then   $H-v$ must appear for $\mathbf A'$ within $t_{\max}-1$ timesteps. By induction hypothesis, this has  probability at most $e_{\sigma'}^{\top} A_{H-v} \vec 1$ for a fixed $\sigma'$. Summing over $\sigma'$ gives a total probability of 
$$
\sum_{\sigma'} A_f[\sigma, \sigma'] e_{\sigma'}^{\top} A_{H-v} \vec 1 = e_{\sigma}^{\top} A_f  A_{H-v} \vec 1 = e_{\sigma}^{\top} A_H \vec 1
$$
as required.

In case (ii), note that by \Cref{ah-prop2} the matrices $A_f$ and $A_H$ commute. Again, if $H$ appears for the original search algorithm $\mathbf A$ within $t_{\max}$ timesteps, then   $H$ must appear for $\mathbf A'$ within $t_{\max}-1$ timesteps. By induction hypothesis, this has  probability at most $e_{\sigma'}^{\top} A_H \vec 1$ for a fixed $\sigma'$. Summing over $\sigma'$ gives a total probability of 
$$
\sum_{\sigma'} A_f[\sigma, \sigma'] e_{\sigma'}^{\top} A_{H} \vec 1 = e_{\sigma}^{\top} A_f  A_{H} \vec 1 = e_{\sigma}^{\top} A_H A_f \vec 1.
$$

Since $A_f$ is substochastic, this is at most $e_{\sigma}^{\top} A_H \vec 1$.

This completes the induction. By countable additivity, we can compute the probability that $H$ ever appears from starting state $\sigma$, as 
$$
\Pr \bigl( \bigvee_{t=0}^{\infty} G_t = H \bigr)= \lim_{t_{\max} \rightarrow \infty} \Pr \bigl( \bigvee_{t=0}^{t_{\max}-1} G_t = H \bigr) \leq  \lim_{t_{\max} \rightarrow \infty} e_{\sigma}^{\top} A_H \vec 1 = e_{\sigma}^{\top} A_H \vec 1
$$
Finally, integrating over $\sigma$ gives $\sum_{\sigma} \mu[\sigma] e_{\sigma}^{\top} A_H \vec 1 =  \mu^{\top} A_H \vec 1$.
\end{proof}

\begin{corollary}
\label{stepscor}
The expected number of steps of the Search Algorithm is at most $ \sum_{H \in \mathfrak S} \mu^{\top} A_H \vec 1$. 
\end{corollary}
\begin{proof}
For each time $t$ that a flaw  $f$ is resampled, the corresponding $G_t$ is an appearing wdag in $\mathfrak W(f)$. These wdags are all  distinct, since on the $k^{\text{th}}$ resampling of $f$ the wdag $G_t$ has exactly $k$ nodes labeled $f$. So by \Cref{witness-tree-lemma_simple}, the expected number of steps is at most \[
\sum_f \sum_{H \in \mathfrak W(f)} \Pr(\text{$H$ appears}) \leq \sum_f \sum_{H \in \mathfrak W(f)} \mu^{\top} A_H \vec 1 = \sum_{H \in \mathfrak S} \mu^{\top} A_H \vec 1. \qedhere
\]
\end{proof}

\subsection{Estimating sums over wdags}
\label{sec-est}
The statement of \Cref{witness-tree-lemma_simple} in terms of matrix products is very general and powerful, but difficult for calculations.  To use it effectively, as for example in \Cref{stepscor}, we need to bound the sums of the form
$$
  \sum_{H}  \mu^{\top} A_H \vec 1
$$
where $H$ ranges over subsets of $\mathfrak W$. 

	There are two, quite distinct, issues here. First, for each individual wdag $H$, we need to estimate $\mu^{\top} A_H \vec 1 $; second, we need to bound the sum of these quantities over $H$. The second issue is well-studied in terms of the probabilistic LLL, but the first issue is not as familiar. To handle it, following \cite{AIS}, we analyze the matrix product via  spectral bounds of the matrices $A_f$.  Let us define a quantity called the \emph{weight} $w(f)$ of each flaw $f$ by\footnote{The work \cite{AIS} uses a more general definition of \emph{charge}, where the ``benchmark'' probability distribution can be different from the initial probability distribution $\mu$. This can be useful in showing convergence of the Search Algorithm for non-commutative resampling oracles. However, this more general definition does not seem useful for distributional properties and parallel algorithms in the context of commutative resampling oracles. Hence, we adopt the simpler definitions here.}
$$
w(f) :=  \max_{\tau \in \Omega}  \frac{\mu^{\top} A_f e_{\tau}}{ \mu(\tau) } 
$$i.e., the maximum possible inflation of a state probability (relative to its probability under $\mu$) incurred by (i) sampling a state $\sigma$ according to $\mu$; (ii) checking that flaw $f$ holds on $\sigma$; and then (iii) resampling flaw $f$ at $\sigma$.

Extending the definition, we define the \emph{weight} of a wdag $H$ by
$$
w(H) = \prod_{v \in H} w(L(v))
$$
\begin{lemma}
  \label{a1}
  For any event $E \subseteq \Omega$ we have $\mu^{\top} A_H e_E \leq  \mu(E) \cdot w(H)$.   
\end{lemma}
\begin{proof}
We can write $A_H = A_{f_1} \cdots A_{f_t}$ where $f_1, \dots, f_t$ are the labels of the nodes of $H$.  By definition, we have $\mu^{\top} A_f \preceq w(f)  \mu^{\top}$ for any $f$. So:
\[
 \mu^{\top} \negthinspace A_H e_E =   \mu^{\top} \negthinspace A_{f_1} \cdots A_{f_t} e_E  \leq \mu^{\top} \negthinspace w(f_1) A_{f_2} \cdots A_{f_t} e_E \leq \dots \leq w(f_1) \cdots w(f_t) \mu^{\top} \negthinspace e_E = w(H) \mu(E).   \qedhere
\]
\end{proof}

We say that a resampling oracle $\mathbf R$ is \emph{regenerating} if $w(f) = \mu(f)$ for all $f$. This perfectly removes the conditional of the resampled flaw.   The original Moser-Tardos algorithm, and extensions to other probability spaces, can be viewed in terms of regenerating oracles \cite{HV}. An equivalent formulation is that  $\mu$ is a left-eigenvector of each matrix $A_f$, with associated eigenvalue $\mu(f)$, i.e.
\begin{equation}
\label{regen-def}
\forall f \qquad \mu^{\top} A_f = \mu(f) \cdot \mu^{\top}
\end{equation}

From \Cref{a1} (applied with $E = \Omega$) and \Cref{witness-tree-lemma_simple} we get the following immediate corollary:
\begin{corollary}
\label{bndcord0}
Any given wdag $H$ appears with probability at most $ w(H)$. 

In particular, if $\mathbf R$ is regenerating, then it appears with probability at most $\prod_{v \in H} \mu(L(v))$ (i.e. the usual Witness Tree Lemma)
\end{corollary}
  
We emphasize that we are not aware of any direct proof of \Cref{bndcord0}; it seems necessary to first show the matrix bound of \Cref{witness-tree-lemma_simple}, and then project down to scalars.  
  
In light of \Cref{a1}, we define for any flaw set $I$ the key quantities
$$
\Psi(I) = \sum_{\substack{H \in \mathfrak W(I)}} w(H), \qquad \qquad \overline \Psi(I) = \sum_{J \subseteq I} \Psi(J).
$$
  We write $\Psi(f) = \Psi( \{f \})$ for brevity.  Note that $\Psi(I) = 0$ if $I$ is not a stable set. A useful and standard formula (see e.g., \cite[Claim 59]{HV}) is that for any set $I$ we have $$
  \Psi(I) \leq \prod_{f \in I} \Psi(f), \qquad \qquad \overline \Psi(I) \leq \prod_{f \in I} (1 + \Psi(f)).
  $$
 We also write $\Psi_{\mathcal F}, \overline \Psi_{\mathcal F}$ to indicate the role of the flaw set $\mathcal F$, if it is relevant.  
  
  We summarize here a few well-known bounds  on these quantities, based on versions of LLL criteria.
\begin{proposition}
\label{main-bound-sum}
 \begin{enumerate}
\item (Symmetric criterion) Suppose that $w(f) \leq p$ and $|\overline \Gamma(f)| \leq d$ for parameters $p,d$ with $e p d \leq 1$.  Then $\Psi(f) \leq e w(f) \leq e p$ for all $f$.
\item (Neighborhood bound) Suppose that every $f$ has $\sum_{g \in \overline \Gamma(f)} w(g) \leq \tfrac{1}{4}$. Then $\Psi(f) \leq 4 w(f)$ for all  $f$.
\item (Asymmetric criterion) Suppose there is some function $x: \mathcal F \rightarrow [0,1)$ with the property that
$$
\forall f \qquad w(f) \leq x(f)  \prod_{g \in \Gamma(f)} (1 - x(g)).
$$
Then $\Psi(f) \leq \frac{x (f)}{1-x(f)}$ for all $f$.

\item (Restricted cluster-expansion) Suppose there is some function $\eta: \mathcal F \rightarrow [0,\infty)$ with the property that
$$
\forall f \qquad \eta(f) \geq w(f) \cdot \sum_{\substack{I \subseteq \overline \Gamma(f) \\ \text{$I$ stable}}} \prod_{g \in I} \eta(g)
$$
Then $\Psi(f) \leq \eta(f)$ for all $f$.

\item (Cluster-expansion) Suppose $\bowtie$ is a symmetric relation on $\mathcal F$ extending $\simeq$, i.e. $f \simeq g \Rightarrow f \bowtie g$, and there is some function $\eta: \mathcal F \rightarrow [0,\infty)$ with the property that
$$
\forall f \qquad \eta(f) \geq w(f) \cdot \sum_{\substack{I \subseteq \mathcal F \\ g \bowtie f \text{ for all $g \in I$}  \\ g_1 \not \bowtie g_2 \text{ for all distinct $g_1, g_2 \in I$}}} \prod_{g \in I} \eta(g)
$$
Then $\Psi(f) \leq \eta(f)$ for all $f$.
\end{enumerate}
\end{proposition}
\begin{proof}
For completeness, we briefly sketch the proofs. 

For the cluster-expansion criterion, first observe that given any wdag $G$, we can topologically sort the nodes of $G$ and then add additional edges between any nodes $v, w$ with $L(v) \bowtie L(w)$. In this way, $\sink(G)$ becomes a set which is independent with respect to the denser dependence relation $\bowtie$. Now use induction on wdag size to show that the total weight of all wdags whose sink nodes (under $\bowtie$) are labeled by $I$ is at most  $\prod_{f \in I} \eta(f)$. Here, we use the fact that if wdag $H$ has sink nodes $v_1, \dots, v_t$ labeled by a stable set $I$, then the sink nodes of $H - v_1 - \dots - v_t$ are labeled by a subset of $\bigcup_{f \in I} \{g: g \bowtie f \}$.

For the restricted cluster-expansion criterion, apply the cluster-expansion criterion with $\bowtie$ being $\simeq$.

For the asymmetric criterion, apply the cluster-expansion criterion with $\eta(f) = \frac{x(f)}{1 - x(f)}$

For the neighborhood bound criterion, apply the asymmetric criterion with $x(f) = 2 w(f)$.

For the symmetric criterion, apply the restricted cluster-expansion criterion with  $\eta(f) = e w(f)$.
\end{proof}

To emphasize the connection between various LLL-type bounds, our analysis of wdags, and the behavior of the Search Algorithm, we record the following results:
\begin{proposition}
\label{bndcord1}
Define the parameter
 $$
 W = \sum_{f \in \mathcal F} \Psi(f) = \sum_{H \in \mathfrak S} w(H)
 $$

The expected number of total resampling is at most $W$.   Moreover, under the conditions of \Cref{main-bound-sum}, we have the following bounds on $W$:
\begin{enumerate}
\item If the symmetric criterion holds, then $W \leq e  \sum_f w(f) \leq O( p |\mathcal F| )$.

\item If the neighborhood-bound criterion holds, then $W  \leq 4  \sum_f w(f) \leq O( |\mathcal F| )$.
\item If the asymmetric criterion holds, then $W  \leq  \sum_f \frac{x(f)}{1-x(f)}$.
\item If the cluster-expansion criterion holds, then $W \leq \sum_f \eta(f)$.
\end{enumerate}
\end{proposition}

\section{Simple distributional  properties}
\label{distrib-sec}
One of the most important consequence of commutativity is that it allows us to bound the distribution of objects generated by the Search Algorithm.   As a warm-up, we will use a construction similar to Section~\ref{WitnessTree-sec} for a ``basic'' distributional bound.  Later, in Section~\ref{sec:detailedist}, we show tighter bounds through a more careful constructions of wdags. 

Consider an event $E \subseteq \Omega$, and let $P(E)$ be the probability that $E$ holds in the output of the Search Algorithm. Define $\check \Gamma(E)$ to be the set of flaws $f$ which can cause $E$ to occur, i.e. there is a transition $\sigma' \notin E \xrightarrow f \sigma \in E$.   Our main goal is to bound $P(E)$; typically, we will seek an upper-bound of the form $P(E) \leq (1 + \eps) \Pr_{\Omega}(E)$, for some small value $\eps > 0$.

If $E$ occurs at some time $t$ in the Search Algorithm, we generate a wdag $G_t$ by initializing $G_t = \emptyset$ and then for each time $s = t-1, \dots, 0$ with resampled flaw $f_s$, modifying $G_t$ according to the following rule:
\begin{itemize}
\item If $f_s \sim G_t$, or $G_t$ has a source node labeled $f_s$, or $f_s \in \check \Gamma(E)$, then prepend $f_s$ to $G_t$. 
\item Otherwise, do not modify $G_t$. 
\end{itemize}

Note that $\sink(G_t) \subseteq \check \Gamma(E)$. We say that wdag $H$ \emph{appears for $E$} if $H$ is isomorphic to $G_t$ for any $t$.

\begin{lemma}
  \label{witness-tree-lemma_simple2}
The probability that a given wdag $H$ appears for $E$ is at most $\mu^{\top} A_H e_E$.
\end{lemma}
\begin{proof}
As in \Cref{witness-tree-lemma_simple}, it suffices to show that if the Search Algorithm runs for at most $t_{\max}$ steps starting with state $\sigma$, where $t_{\max}$ is an arbitrary integer, then $H$ appears with probability at most $e_{\sigma}^{\top} A_H e_E.$ We prove this claim by induction on $t_{\max}$.

If $H = \emptyset$ and $\sigma \in E$, then $e_{\sigma}^{\top} A_H e_E = 1$ and the bound holds vacuously. This is the only way that $H$ can appear if $t_{\max} = 0$. So, for the induction step, suppose that $t_{\max} > 0$ and either $H \neq \emptyset$ or $\sigma \not \in E$.  Then $G_0 \neq \emptyset$, and the only way for $H$ to appear is to have $G_t = H$ for $t \in \{1, \dots,  t_{\max}  \}$. Now suppose $\mathbf S$ selects a flaw $f$ to resample in $\sigma$. We can view the evolution of the Search Algorithm $\mathbf A$ as a two-part process: first resample $f$, reaching a state $\sigma'$; then execute a new search algorithm $\mathbf A'$ starting at state $\sigma'$.

 As in \Cref{witness-tree-lemma_simple}, in order for $H$ to appear, one of the following two mutually exclusive conditions must hold: (i) $H$ has a unique source node $v$ labeled $f$ and $G'_{t-1} = H - v$; or  (ii) $f \not \sim H$ and $f \notin \check \Gamma(E)$ and $G'_{t-1} = H$. 

In case (i),  $H - v$ would appear for $E$ in search algorithm $\mathbf A'$ within $t_{\max}-1$ timesteps. By induction hypothesis, this has  probability at most $e_{\sigma'}^{\top} A_{H-v} e_E$ for a fixed $\sigma'$. Summing over $\sigma'$ gives a total probability of  $\sum_{\sigma'} A_f[\sigma, \sigma'] e_{\sigma'}^{\top} A_{H-v} e_E = e_{\sigma}^{\top} A_f  A_{H-v} e_E = e_{\sigma}^{\top} A_H e_E$.

In case (ii),  $H $ would appear for $E$ in search algorithm $\mathbf A'$ within $t_{\max}-1$ timesteps. By induction hypothesis, this has  probability at most $e_{\sigma'}^{\top} A_H e_E$ for a fixed $\sigma'$. Summing over $\sigma'$ gives a total probability of  $\sum_{\sigma'} A_f[\sigma, \sigma'] e_{\sigma'}^{\top} A_{H} e_E = e_{\sigma}^{\top} A_f  A_{H} e_E$. Since $A_f$ commutes with $A_H$, this equals $e_{\sigma}^{\top} A_{H} A_f e_E$. Since $f \notin \check \Gamma(E)$, resampling $f$ can never cause $E$ to occur, and so $e_{\tau}^{\top} A_f e_E = 0$ for any state $\tau \notin E$; equivalently, we have $A_f e_E \preceq e_E$. So, overall, the probability is at most $e_{\sigma}^{\top} A_{H} e_E$ as claimed. 
\end{proof}

This gives us the following crisp bound:
 
\begin{theorem}
\label{crispdist}
$P(E) \leq \mu(E) \overline{\Psi}(\check \Gamma(E))$.
\end{theorem}
\begin{proof}
If $E$ is ever true, then some wdag $H$ appears for $E$, where necessarily $\sink(H) \subseteq \check \Gamma(E)$. A union bound over such wdags gives $P(E) \leq \sum_{I \subseteq \check \Gamma(E)} \sum_{H \in \mathfrak W(I)} \mu^{\top} A_H e_E$. By \Cref{a1}, this is at most $\sum_{I \subseteq \check \Gamma(E)} \sum_{H \in \mathfrak W(I)} w(H) \mu(E) = \mu(E) \overline{\Psi} (\check \Gamma(E))$.
\end{proof}

We can combine this with common LLL criteria to obtain more readily usable bounds; the proofs are immediate from bounds on $\Psi$ shown in \Cref{main-bound-sum}.
\begin{proposition}
\label{main-bound-sum2}
Under four criteria of \Cref{main-bound-sum}, we have the following estimates for $P(E)$:
\begin{enumerate}
\item If the symmetric criterion holds,  then  $P(E) \leq \mu(E) \cdot e^{e |\check \Gamma(E)|  p}$.

\item If the neighborhood-bound criterion holds, then $P(E) \leq \mu(E) \cdot e^{4 \sum_{f \in \check \Gamma(E)} w(f)}$.

\item If function $x$ satisfies the asymmetric criterion, then $P(E) \leq \mu(E) \cdot \prod_{f \in \check \Gamma(E)} \frac{1}{1 - x(f)}$.

\item If function $\eta$ satisfies the cluster-expansion criterion,  then $P(E) \leq \mu(E) \cdot \sum_{\substack{ I \subseteq \check \Gamma(E) \\ \text{$I$ stable}}}  \prod_{g \in I} \eta(g)$.
\end{enumerate}
\end{proposition}

We remark that Iliopoulos \cite{LLLWTL} had shown a bound similar to \Cref{crispdist}, but with three additional technical restrictions: (i) it requires strong commutativity; (ii) it requires the construction of a commutative resampling oracle for the event $E$ itself; and (iii)  it gives a strictly worse bound for non-regenerating oracles.

\section{Parallel algorithms} 
\label{sec:parallel}
Moser \& Tardos \cite{MT} described a simple parallel version of their resampling algorithm:
\begin{algorithm}[H]
\centering
\begin{algorithmic}[1]
\State Draw state $X$ from distribution $\mu$
\While{some bad-event is true on $X$}
\State{Select some arbitrary maximal independent set (MIS) $I$ of bad-events true on $X$} 
\State{Resample, in parallel, all variables $X_i$ involved in events in $I$}
\EndWhile
\end{algorithmic}
\caption{Parallel Moser-Tardos algorithm}
\label{mt-par-alg}
\end{algorithm}

This algorithm depends heavily on the properties of the variable LLL, as it requires that bad-events which are independent in the LLL sense are also computationally independent.  Parallel algorithms have been developed for a number of other probability spaces \cite{PermHarris, LLLLBeyond}, with intricate and ad-hoc analyses. Based on commutativity, Kolmogorov \cite{Kolmofocs} proposed a generic framework which we summarize as follows:
\begin{algorithm}[H]
\centering
\begin{algorithmic}[1]
\State Draw state $\sigma$ from distribution $\mu$
\While{some flaw holds on $\sigma$}
\State{Set $V$ to be the set of flaws currently holding on $\sigma$}
\While{$V \neq \emptyset$}
\State Select, arbitrarily, a flaw $f \in V$.
\State Update $\sigma \leftarrow \mathbf R_{\sigma} (\sigma)$.
\State Remove from $V$ all flaws $g$ such that $\sigma \notin g$ or $f \simeq g$
\EndWhile
\EndWhile
\end{algorithmic}
\caption{Generic parallel resampling framework}
\label{kolm-alg}
\end{algorithm}

We emphasize that this is a \emph{sequential} algorithm, which can be viewed as a version of the Search Algorithm with an unusual flaw-selection rule. Each iteration of the main loop (lines 3 -- 7) is called a \emph{round}.  Kolmogorov showed that if $\mathbf R$ is strongly commutative, then Algorithm~\ref{kolm-alg} terminates after polylogarithmic rounds with high probability. Harris \cite{ParallelHarris} further showed that if $\mathbf R$ satisfies a property called \emph{obliviousness} (see Section~\ref{comp-sec}), then each round can be simulated in polylogarithmic time. These two results combine to give an overall RNC search algorithm.   Most known parallel local search algorithms, including  Algorithm~\ref{mt-par-alg}, fall into this framework.

 We  will extend these results to the matrix-commutative setting. For $k = 1, 2, \dots$, define $V_k$ to be the set of flaws $V$ in round $k$, and define $I_k$ to be the set of flaws which are actually resampled at round $k$ (i.e. a flaw $f$ selected at some iteration of line 5). Note that each $I_k$ is a stable set and $I_k \subseteq V_k$. 
 
  Let $b_k = \sum_{i < k} |I_i|$ be the total number of resamplings made before round $k$; thus $b_1 = 0$, and when we ``serialize'' Algorithm~\ref{kolm-alg} and view it as an instance of the Search Algorithm, the resamplings in round $k$ of Algorithm~\ref{kolm-alg} correspond to the resamplings at times $b_{k}, \dots, b_{k+1} - 1$ of the Search Algorithm.
\begin{proposition}
\label{gaha0}
 For each $f \in V_k, k \geq 2$ there exists $g \in I_{k-1}$ with $f \simeq g$.
\end{proposition}
\begin{proof}
First, suppose that $f \notin V_{k-1}$. In this case, by \Cref{cause-obs}, the only way $f$ could become true at round $k$ would be that some $g \sim f$ was resampled at round $k-1$, i.e. $g \in I_{k-1}$. Otherwise, suppose that $f \in V_{k-1}$. Then either it was removed from $V_{k-1}$ due to resampling of some $g \simeq f$, or $f$ became false during round $k-1$.  In the latter case, since it later become true at the beginning of round $k$, some other $g'  \sim f$ was resampled in round $k-1$ after $g$.
\end{proof}

For a node $v$ in a wdag $G$, we define the \emph{depth} of $v$ to be the length of the longest path to any sink node; we define the \emph{depth} of $G$ to be the maximum depth of its vertices.

\begin{proposition}
\label{gaha1}
For each $t \in \{ b_k, \dots, b_{k+1} - 1 \}$ the corresponding wdag $G_t$ as constructed in Section~\ref{wdag-construct} has depth precisely $k$.
\end{proposition}
 \begin{proof}
 Fix $t$, and consider the partial construction of $G_t$ by adding nodes backwards in time for each flaw $f_s$ resampled at times $s = t, t-1, \dots, b_j$; let $H_j$ be the resulting wdag. So $H_1 = G_t$. We show by backwards induction on $j$ that each $H_j$ has depth precisely $k - j + 1$, and the nodes $v \in H_j$ with depth $k - j + 1$ correspond to resamplings in round $j$. 
 
 The base case $j = k$ is clear, since then $H_j$ consists of a singleton node corresponding to the resampling at time $t$ in round $k$.

For the induction step,  observe that $H_j$ is formed from $H_{j+1}$ by adding nodes corresponding to resamplings in $I_j$. By induction hypothesis, $H_{j+1}$ has depth $k-j$. Since $I_j$ is a stable set, we have $\depth(H_{j}) \leq 1 + \depth(H_{j+1}) =  k-j+1$ and furthermore the nodes at maximal depth correspond to resamplings in $I_j$. So we just need to show that there is at least one such node.

Consider any node $v$ of $H_{j+1}$ with depth $k-j$; by induction hypothesis this corresponds to a resampling in round $j+1$. By \Cref{gaha0}, we have $L(v) \simeq f_{s}$ for some time $s$ in round $j$. The procedure for generating $G_t$ will thus add a node labeled $f_s$ with an edge to $v$; this node has depth $k-j+1$ in $H_j$. 

This completes the induction.  The stated bound then holds since $G_t = H_1$.
\end{proof}

\begin{proposition}
\label{gaha11}
For any flaw $f$ and $k \geq 1$, we have $\Pr( f \in V_k) \leq \sum_{\substack{H \in \mathfrak W(f) \\ \depth(H) = k}} \mu^{\top} A_H \vec 1$.
\end{proposition}
\begin{proof}
As discussed, Algorithm~\ref{kolm-alg} can be viewed as an instantiation of the Search Algorithm with a certain flaw selection rule $\mathbf S$. For fixed $f$ and $k$, consider a new flaw selection rule $\mathbf S_{f,k}$, which agrees with $\mathbf S$ up to round $k$, and then selects $f$ to resample at round $k$ if it is true. The behavior of the Search Algorithm for $\mathbf S$ and $\mathbf S_{f,k}$ is identical up through the first $b_{k}$ resamplings. Subsequently, we have $f \in V_k$ if and only if the Search Algorithm with $\mathbf S_{f,k}$ selects $f$ at time $t = b_{k}$. In this case, by \Cref{gaha1}, the corresponding wdag $G_t \in \mathfrak W(f)$ would have depth $k$.

 So we can bound the probability of $f \in V_k$ by a union bound over such wdags $H \in \mathfrak W(f)$ of depth $k$ and applying \Cref{witness-tree-lemma_simple}.
\end{proof}

As is usual for the parallel LLL, we analyze the runtime by using an ``inflated'' weight function for some $\eps > 0$ defined as
$$
w_{\epsilon}(H) = w(H) (1+\epsilon)^{|H|} = \prod_{v \in H}  (1+\epsilon) w(L(v))  \qquad \qquad 
W_{\epsilon} = \sum_{H \in \mathfrak S} w_{\epsilon}(H)
$$

Bounding $W_{\epsilon}$ is very similar to bounding $W = W_0$ with a small ``slack'' in the weights. 
With this definition, we get the following bounds:
\begin{proposition}
\label{gaha33}
\begin{enumerate}
\item  $\sum_k \bE[|V_k|] \leq W$.
\item  For any integer $t \geq 1$ and any $\eps > 0$, the probability that Algorithm~\ref{kolm-alg} runs for more than $2 \ell$ rounds is at most $(1+\eps)^{-\ell} W_{\eps} / \ell$. 
\item For $\eps, \delta \in (0,\tfrac{1}{2})$, Algorithm~\ref{kolm-alg} terminates in $O( \frac{\log(1/\delta + \epsilon W_\epsilon)}{\epsilon})$ rounds with probability at least $1 - \delta$.
 \end{enumerate}
\end{proposition}
\begin{proof}
First, by \Cref{a1} and \Cref{gaha11}, we have
$$
\bE[|V_k|] \leq \sum_f \sum_{\substack{H \in \mathfrak W(f), \\ \depth(H) = k}} w(H) = \sum_{\substack{H \in \mathfrak S, \\ \depth(H) = k} } w(H),
$$
so $\sum_k \bE[|V_k|] \leq \sum_{\substack{H \in \mathfrak S}} w(H) = \sum_f \Psi(f) = W$. 

For the second claim, consider random variable $Y = \sum_{k \geq \ell} |V_k|$. If Algorithm~\ref{kolm-alg} reaches iteration $2 \ell$, then necessarily $V_k \neq \emptyset$ for $k = \ell, \dots, 2 \ell$, and so $Y \geq \ell$. By  Markov's inequality applied to $Y$, we thus get
$$
\Pr( \text{Alg reaches round $2 \ell$}) \leq \Pr( Y \geq \ell ) \leq \bE[Y]/ \ell  \leq \sum_{H \in \mathfrak S, \depth(H) \geq \ell} w(H)/\ell. 
$$

We can then calculate
\[
 \sum_{\substack{H \in \mathfrak S, \depth(H) \geq \ell}} w(H)  =   \sum_{\substack{H \in \mathfrak S, \depth(H) \geq \ell}} w_{\epsilon} (H) (1+\epsilon)^{-|H|}    \leq   (1+\epsilon)^{-\ell}  \sum_{\substack{H \in \mathfrak S}} w_{\epsilon} (H) =  (1+\epsilon)^{-\ell} W_{\epsilon}.
\]

The third claim follows from the second claim via Markov's inequality.
\end{proof}

 Using standard estimates (see \cite{ParallelHarrisHaeupler, Kolmofocs, AIS}) we get the following bounds:
\begin{proposition}
\begin{enumerate} 
\item If the resampling oracle is regenerating and the probability vector $p (1+\epsilon)$ satisfies the LLLL criterion, then $W_{\epsilon/2} \leq O(m/\epsilon)$ and Algorithm~\ref{kolm-alg} terminates after $O( \frac{ \log (m/\delta)}{\epsilon})$ rounds with probability $1 - \delta$.
\item If $w(f) \leq p$ and $|\overline \Gamma(f)| \leq d$ for alls flaws $f$, where $e p d (1+\epsilon) \leq 1$, then $W_{\epsilon/2} \leq O(m/\epsilon)$ and Algorithm~\ref{kolm-alg} terminates after $O( \frac{ \log (m/\delta)}{\epsilon})$ rounds with probability at least $1 - \delta$.
\item If the resampling oracle is regenerating and oblivious and satisfies the computational requirements of \cite{ParallelHarris} for input length $n$, then with probability $1 - 1/\poly(n)$ the algorithm of \cite{ParallelHarris} terminates in $O( \frac{ \log^4(n + \epsilon W_\epsilon) }{\epsilon})$ time on an EREW PRAM.
 \end{enumerate}  
\end{proposition}

\section{Compositional properties for resampling oracles} 
\label{comp-sec}
In many settings, the flaws and resampling oracles are built out of simpler, ``atomic'' events. In \cite{ParallelHarris}, Harris described a generic construction, and generic parallel algorithm, when the atomic events satisfy an additional property called \emph{obliviousness}. Let us now review this construction, and how it works with commutativity. 

Consider a set $\mathcal A$ of events, and a symmetric dependence relation $\sim$. It is allowed, but not required, to have $f \sim f$ for $f \in \mathcal A$. We refer to the elements of $\mathcal A$ as \emph{atoms}. These should be thought of as ``pre-flaws'', that is, they have the \emph{structural}  properties of a resampling oracle, but do not necessarily satisfy any convergence condition such as the LLLL.  

The obliviousness definition in \cite{ParallelHarris} can be formulated as follows:

\begin{definition}[Explicitly-oblivious resampling oracle \cite{ParallelHarris}]
  \label{obliv-def}

  The resampling oracle $\mathbf R$ is \emph{explicitly-oblivious} if each $\mathbf R_f$ can be implemented by drawing a random seed $r$ from a set $R_f$ and setting $\sigma' = \mathbf R_f(\sigma) =  F_f( \sigma, r)$ for a \emph{deterministic} function $F_f$. Furthermore, for each pair $f, g \in \mathcal A$ with $f \not \sim g$, there is a set $R_{f;g} \subseteq R_f$ such that $$
(F_f(\sigma, r) \in g ) \Leftrightarrow (\sigma \in f \cap g \wedge  r \in R_{f;g})
  $$
\end{definition}

We list a number of probability spaces with this property; see \cite{ParallelHarris} for further details and proofs.
\begin{enumerate}
\item \textbf{Variables:} This is easy. The probability space $\Omega$ has $n$ independent variables $X_1, \dots, X_n$. For each pair $(i,y)$, we have an atom $f_{iy} = \{ X_{i}  = y \}$. We have $f_{iy} \sim f_{i'y'}$ iff $i = i', y \neq y'$. The space $R_{f_{iy}}$ is defined by drawing a value $y'$ from the distribution of $X_i$, and setting $F_f(X, y') = (X_1, \dots, X_{i-1}, y', X_{i+1}, \dots, X_n)$.

\item \textbf{Permutations:} $\Omega$ is the uniform distribution on permutations $\pi$ on $\{1, \dots, n \}$.  For each pair $(x,y)$, we have an atom $f_{xy} = \{ \pi x  = y \}$. We have $f_{xy} \sim f_{x'y'}$ iff $x = x', y \neq y'$ or $x \neq x', y = y'$. The space $R_{f_{xy}}$ is defined by choosing a value $z$ uniformly from $ \{1, \dots, n \}$, and setting $F(\pi, z) = (y \ z) \pi$. (Here $(y \ z)$ is the permutation which swaps $y$ with $z$.)

\item \textbf{Clique perfect matchings:}  $\Omega$ is the uniform distribution on perfect matchings $M$ of the $n$-clique, where $n$ is even. For each pair $(x,y)$, we have an atom $f_{xy} = \{ \{x, y \} \in M \}$. Note that $f_{xy} = f_{yx}$.  We have $f_{xy} \sim f_{x'y'}$ iff $| \{ x, y \} \cap \{x', y' \} |= 1 $. For $x < y$, the space $R_{f_{xy}}$ is defined by choosing a value $z$ uniformly from $\{1, \dots, n \} \setminus \{x \}$, and setting $F(M, z) = (y \ z) M$ (with the natural left-group action of permutations on matchings).

\item \textbf{Hamiltonian cycles:} $\Omega$ is the uniform distribution on $n$-cycle permutations.  For each sequence of distinct elements $\vec x = x_1, \dots, x_k$ there is an atom $f_{x} = \{ \pi x_i = x_{i+1}: i = 1, \dots, k-1 \}$. We have $f_{x} \sim f_{x'}$ iff $\{x_1, \dots, x_k \} \cap \{x'_1, \dots, x'_{k'} \} \neq \emptyset$. The resampling oracle for this space is too complicated to explain here.
\end{enumerate}

While this definition of obliviousness is critical for the parallel algorithm, we use a simpler and more general matrix-based notion.

\begin{definition}[Matrix-oblivious resampling oracle]
The resampling oracle $\mathbf R$ is \emph{matrix-oblivious} if for each stable set $C$ and atom $f \not \sim C$ there holds $ A_f e_{\langle C \rangle} \propto e_{f \cap \langle C \rangle}$.

 (Recall that $\langle C \rangle$ is the intersection of the atoms in $C$.)
\end{definition}

\begin{proposition}
\label{obl-p0}
If $\mathbf R$ is explicitly-oblivious but not necessarily commutative, then $\mathbf R$ is matrix-oblivious.
  \end{proposition}
  \begin{proof}
  Let $C = \{g_1, \dots, g_k \}$ and define $g = \langle C \rangle$. Observe that, when implementing $\sigma' = \mathbf R_f(\sigma)$ by drawing a seed $r$ from $R_f$, we have $\sigma' \in g$ if and only if $r \in R_{f;g_1} \cap \dots \cap R_{f;g_k}$ and $\sigma \in f \cap g_1 \cap \dots \cap g_k = f \cap g$.  In particular, $e_{\sigma}^{\top} A_f e_g$ is constant for $\sigma \in f \cap g$ and hence $A_f e_g \propto e_{f \cap g}$.  
  \end{proof}
  
\textbf{For the remainder of this paper, we say \emph{oblivious} to mean \emph{matrix-oblivious}.}  

Let us suppose now that $\mathbf R$ is oblivious.   From $\mathcal A$, one can construct an enlarged set of events 
$$
\oa = \{ \langle C \rangle \mid \text{$C$ a stable set of $\mathcal A$} \}.
$$
We define the relation $\sim$ on $\oa$ by setting $\langle C \rangle \sim \langle C' \rangle$ iff $C \sim C'$. For each event $f = \langle C \rangle \in \oa$, we define a corresponding resampling oracle $\obr_f$ as follows. Choose some arbitrary enumeration $C = \{ f_1, \dots, f_t \}$. Given a state $\sigma_1$, apply $\mathbf R_{f_1}$ repeatedly until the resulting state $\sigma_2$ is in $f_2 \cap \dots \cap f_t$ (i.e. via rejection sampling). Then, again apply $\mathbf R_{f_2}$ repeatedly until the state is in $f_3 \cap \dots \cap f_t$, and so on.

 The intent is to choose the flaw set $\mathcal F$ to be some arbitrary subset of $\oa$;  as before, $\oa$ does not necessarily satisfy any LLLL convergence criterion.

It would seem reasonable that if $\mathbf R$ is commutative, then $\obr$ should be as well.  We will indeed show this for matrix commutativity. By contrast, it does not seem to hold for strong commutativity. This is a good illustration of how the new commutativity definition is easier to work with, beyond its advantage of greater generality.

\begin{proposition}
\label{obl-p00}
Suppose $\mathbf R$ is oblivious and regenerating, but not necessarily commutative. Then for a stable set $C$ and atom $f \not \sim C$ there holds $A_f e_{\langle C \rangle} =\frac{\mu(f) \mu(\langle C \rangle)}{\mu(f \cap \langle C \rangle)} e_{f \cap \langle C \rangle}$.
\end{proposition}
\begin{proof}
By hypothesis, we have $A_f e_{\langle C \rangle} = p \cdot e_{f \cap \langle C \rangle}$ for scalar $p$. Since $\mathbf R$ is regenerating, we have on the one hand $\mu^{\top} A_f e_{\langle C \rangle} = \mu(f) \mu^{\top} e_{\langle C \rangle} = \mu(f) \mu(\langle C \rangle)$, and on the other hand $\mu^{\top} e_{f \cap \langle C \rangle} = \mu(f \cap \langle C \rangle)$.
\end{proof}
  
\begin{proposition}
  \label{obl-p1}
  Suppose $\mathbf R$ is oblivious but not necessarily commutative. For $f = \langle C \rangle$ in $\oa$, suppose that we have fixed an enumeration $C = \{f_1, \dots, f_t \}$ to define $\obr_f$.    Then $
  A_f \propto A_{f_1} \cdots A_{f_t}$. If $\mathbf R$ is regenerating, then in particular $ A_f = \frac{\mu(f) }{\mu(f_1) \cdots \mu(f_t)}  A_{f_1} \cdots A_{f_t}$.
\end{proposition}
\begin{proof} We show it by induction on $t$. The case $t = 1$ is immediate.   For $t > 1$, let $f' = \langle \{ f_2, \dots, f_t \} \rangle$. We can view $\obr_f(\sigma)$ recursively:  first resample $f_1$, conditional on the resulting state $\sigma'$ being in $f'$; then apply $\mathbf R_{f'}$. So $ e_{\sigma}^{\top} A_f = \frac{ e_{\sigma}^{\top} A_{f_1} A_{f'} }{e_{\sigma}^{\top} A_{f_1} e_{f'}}$. By \Cref{obl-p0} we have $e_{\sigma}^{\top} A_{f_1} e_{f'} \propto e_{\sigma}^{\top} e_{f_1 \cap f'} = 1$ since $\sigma \in f_1 \cap f' = f$.  Overall, we have $A_f \propto A_{f_1} A_{f'}$, and the first result then follows immediately from induction.

The proof for the case where $\mathbf R$ is regenerating is completely analogous, where we use \Cref{obl-p00} to determine $e_{\sigma}^{\top} A_{f_1} e_{f'} = \mu(f_1) \mu(f') / \mu(f)$ for $\sigma \in f$.
\end{proof}

\begin{theorem}
  Suppose  $\mathbf R$  is oblivious, but not necessarily commutative. 
  \begin{itemize}
  \item $\obr$ with dependence relation $\sim$ is an oblivious resampling oracle for $\oa$.
  \item If $\mathbf R$  is regenerating on $\mathcal A$, then $\obr$ is regenerating on $\oa$.
  \end{itemize}
\end{theorem}
\begin{proof}
Consider $f = \langle \{ f_1, \dots, f_t \} \rangle \in \oa$ and a stable set $C = \{ \langle C_1 \rangle, \dots, \langle C_k \rangle \}$ of $\oa$ with $f \not \sim C$. Let $g = \langle C \rangle$; note that also $g = \langle C_1 \cup \dots \cup C_k \rangle$, where $C_1 \cup \dots \cup C_k$ is a stable set of $\mathcal A$.  By \Cref{obl-p1}, we have $A_f e_g \propto A_{f_1} \cdots A_{f_t} e_g$. By matrix-obliviousness of $f_t$, this is proportional to $A_{f_1} \dots A_{f_{t-1}} e_{f_t \cap g}$. Repeatedly again applying the definition of matrix-oblivious to atoms $f_{t-1}, \dots, f_1$ gives us us $A_f e_g \propto e_{f_1 \cap \dots \cap f_t  \cap g} = e_{f \cap g}$ as required for the first claim. Similarly, if $\mathbf R$ is regenerating, then \Cref{obl-p1} gives $
\mu^{\top} A_g = \frac{ \mu^{\top} \mu(g) }{\mu(f_1) \cdots \mu(f_t)} A_{f_1} \cdots A_{f_t} = \mu(g) \mu^{\top}
$
since $\mu^{\top} A_{f_i} = \mu(f_i)$ for each $i$.
\end{proof}
  
  \begin{theorem}
  \label{obl-p2}
  Suppose  $\mathbf R$ is commutative and oblivious. Then the following properties hold for $\obr$:
  \begin{enumerate}
  \item For a stable set $I = \{ \langle C_1 \rangle, \dots, \langle C_k \rangle \}$ in $\oa$, the multiset union $J = C_1 \uplus C_2 \uplus \dots \uplus C_k$ (i.e. the number of copies of $f$ is the sum of those in $C_1, \dots, C_k)$ is a stable multiset of $\mathcal A$, with $A_I \propto A_J$. 
  \item The matrix $A_f$ for $f = \langle C \rangle$  in $\oa$ does not depend on the chosen enumeration $C = \{ f_1, \dots, f_t \}$.  
  \item $\overline{\mathbf R}$ is commutative on $\oa$.
  \end{enumerate}
\end{theorem}
\begin{proof}
\begin{enumerate}
\item  Since $I$ is stable in $\oa$, we have $C_i \not \sim C_j$ for all pairs $i,j$. So $J$ is a stable multiset. Furthermore, by \Cref{obl-p1}, we have we have $A_J = \prod_{i=1}^k A_{\langle C_i \rangle} \propto  \prod_{i=1}^k \prod_{f \in C_i} A_f = \prod_{f \in J} A_f = A_J$.
\item  By \Cref{obl-p1}, we have $A_f = c A_{f_1} \cdots A_{f_t}$ for a scalar $c$. Since the matrices $A_{f_i}$ all commute, the RHS does not depend on the enumeration of $C$. Furthermore,  $c$ can be determined from $A_{f_1} \cdots A_{f_t}$ by choosing an arbitrary state $\sigma \in f$ and setting $c = \frac{1}{ e_{\sigma}^{\top} A_{f_1} \cdots A_{f_t} \vec 1}$.
\item  Let $g = \langle C \rangle, g' = \langle C' \rangle$ with $g \not \sim g'$. So $f \not \sim f'$ for all $f \in C, f' \in C'$.  By \Cref{obl-p1} we have
  $$
  A_g A_{g'} = c_g c_{g'} \Bigl( \prod_{f \in C} A_{f} \prod_{f' \in C'} A_{f'} \Bigr), \qquad \qquad  A_{g'} A_g = c_{g'} c_g  \Bigl(  \prod_{f' \in C'} A_{f'} \prod_{f \in C} A_{f} \Bigr)
  $$
  for scalar constants $c_g, c_{g'}$. All these matrices $A_f, A_{f'}$  commute, so both quantities are equal. \qedhere
  \end{enumerate}
\end{proof}

As an example of an LLL construction using atomic events, recall the latin transversal application. Here, we have an $n \times n$ colored array $C$, where each color appears at most $\Delta$ times. We seek a permutation $\pi$ where all colors $C(i, \pi)$ are distinct.  
  
\begin{proposition}
  \label{latin-prop}
  Suppose that $\Delta = \frac{27}{256} n$. Then the expected number of steps of the Search Algorithm is $O(n)$. Furthermore,  we have $\Psi(f) \leq \frac{256}{81 n^2}$ for each flaw $f$.
\end{proposition}
\begin{proof}
 We start with the atomic set $\mathcal A$ for the permutation setting, and then build the set of flaws $\mathcal F$ as a subset of $\oa$. For each pair of cells $(x_1, y_1), (x_2, y_2)$ of the same color, we have a flaw $f = \langle \{ f_{x_1 y_1}, f_{x_2 y_2} \} \rangle$, i.e. that $\pi x_1 = y_1 \wedge \pi x_2 = y_2$.
 
We apply the cluster-expansion criterion with $\eta(f) = \frac{256}{81 n^2}$ for each flaw $f$ and define the $\bowtie$ relation by setting $f \bowtie g$ if $f, g$ overlap in a row or column. This indeed extends $\simeq$, which has $f \sim g$ if they overlap on a coordinate and \emph{disagree} on the corresponding other coordinate.

Consider now a flaw $f$ corresponding to cells $(x_1, y_1), (x_2, y_2)$, and a corresponding set $I$ of $\bowtie$-neighbors. At most one flaw $g \in I$ can overlap with column $x_1$ (any two such elements $g_1, g_2$ would have $g_1 \bowtie g_2$); given $x_1' = x_1$, there are $n$ choices for $y_1'$, then given the pair $(x_1', y_1')$, there are at most $\Delta - 1$ other cells with the same color. Similar arguments apply to elements in $I$ overlapping the other rows and columns. Overall, the sum of $\prod_{g \in I} \eta(g)$ over all such sets $I$ is at most $(1 + n (\Delta-1)  \frac{256}{81 n^2})^4$. So we need to show that
$$
\frac{256}{81 n^2} \geq \frac{1}{n^2} \cdot \bigl( 1 + n (\Delta-1)  \frac{256}{81 n^2} \bigr)^4
$$
which is a routine calculation for $n \geq 2$. 

The total number of flaws is at most $n^2 (\Delta-1) / 2 = O(n^3)$. So $W \leq |\mathcal F| \cdot \frac{256}{81 n^2} \leq O(n)$.
\end{proof}
  
\section{More detailed distributional bounds}
\label{sec:detailedist}
As before, consider an event $E$ in $\Omega$, and let $P(E)$ be the probability that $E$ holds in the output of the Search Algorithm.  We will develop tighter bounds on $P(E)$, via a more refined construction of wdags. Namely, suppose that $E$ holds at some time $t$. We build a corresponding wdag $J_t$ by initializing $J_t = \emptyset$ and then, for each time $s = t-1, \dots, 0$ with resampled flaw $f_s$, updating $J_t$ as follows:
\begin{itemize}
\item If $A_{f_s} A_{J_t} e_E \not \preceq A_{J_t} e_E$,  or if $J_t$ has a source node labeled $f_s$, then prepend $f_s$ to $J_t$
\item Otherwise, do not modify $J_t$.
\end{itemize}

We say that wdag $H$ \emph{appears for $E$} if $J_t$ is isomorphic to $H$ for any time $t \geq 0$.  (This overrides the definition in \Cref{distrib-sec}.)

\begin{lemma}
\label{jtappear}
Any given wdag $H$ appears for $E$ with probability at most $\mu^{\top} A_H e_E$.
\end{lemma}
\begin{proof}
As in \Cref{witness-tree-lemma_simple}, it suffices to show that if the Search Algorithm runs for at most $t_{\max}$ steps starting with state $\sigma$, where $t_{\max}$ is an arbitrary integer, then $H$ appears for $E$ with probability at most $e_{\sigma}^{\top} A_H e_E.$ We prove this claim by induction on $t_{\max}$. 

If $H = \emptyset$ and $\sigma \in E$, then $e_{\sigma}^{\top} A_H e_E = 1$ and the bound holds vacuously. This is the only way that $H$ can appear if $t_{\max} = 0$. So, for the induction step, suppose that $t_{\max} > 0$ and either $H \neq \emptyset$ or $\sigma \not \in E$. Then $J_0 \neq H$, and the only way for $H$ to appear for $E$ is to have $J_t = H$ for some $t \in \{1, \dots, t_{\max} \}$. Now suppose $\mathbf S$ selects a flaw $f$  in $\sigma$. We view the evolution of the Search Algorithm $\mathbf A$ as a two-part process: we first resample $f$, reaching state $\sigma'$ with probability $A_{f}[ \sigma, \sigma']$. We then execute a new search algorithm $\mathbf A'$ starting at state $\sigma'$. 

So in order for $H$ to appear, one of the following two mutually exclusive conditions must hold:  (i) $H$ has a unique source node $v$ labeled $f$ and $J'_{t-1} = H - v$; or  (ii) $H$ has no such node and $J'_{t-1} = J_t = H$ and $A_f A_H e_E \preceq A_H e_E$.

In case (i), $H - v$ would appear for $E$ in search algorithm $\mathbf A'$ within $t_{\max} - 1$ timesteps. By induction hypothesis, this has  probability at most $e_{\sigma'}^{\top} A_H e_E$ for fixed $\sigma'$. Summing over $\sigma'$ gives a total probability of 
$$
\sum_{\sigma'} A_f[\sigma, \sigma'] e_{\sigma'}^{\top} A_{H-v} e_E = e_{\sigma}^{\top} A_f  A_{H-v} e_E = e_{\sigma}^{\top} A_H e_E
$$ 

In case (ii), $H $ would appear for $ E$ in search algorithm $\mathbf A'$ within $t_{\max} - 1$ timesteps. By induction hypothesis, this has  probability at most $e_{\sigma'}^{\top} A_H e_E$ for fixed $\sigma'$. Summing over $\sigma'$ gives a total probability of 
$$
\sum_{\sigma'} A_f[\sigma, \sigma'] e_{\sigma'}^{\top} A_{H} e_E = e_{\sigma}^{\top} A_f  A_{H} e_E
$$

Since $A_f A_H \preceq A_H$, this is at most $e_{\sigma}^{\top} A_H e_E$, again completing the induction.
\end{proof}

Let us say that a time $t$ is \emph{good for $E$} if $E$ is true at time $t$, and either (i) $t = 0$ (i.e. the initial sampling of the variables) or (ii) $t > 0$ and $E$ is false at time $t-1$.   For a subset $X \subseteq \Omega$, we say that $t$ is \emph{good for $X,E$} if $t$ is good for $E$ and $X$ is true at time $0, \dots, t-1$.  Define $N(X,E)$ to be the expected number of times $t$ that are good for $X,E$. Correspondingly, define $\mathfrak J[X,E]$ to be the collection of wdags which can appear for $E$ at such times $t$. We have the following main characterization:
\begin{proposition}
\label{N-prop}
There holds $N(X, E) \leq \sum_{H \in \mathfrak J[X,E]} \mu^{\top} A_H e_E$.
\end{proposition}
\begin{proof}
We claim that, if $E$ is true at times $t_1, t_3$ and false at time $t_2$, where $t_1 < t_2 < t_3$, then $J_{t_1} \neq J_{t_3}$.  We show this by induction on $t_1$. First suppose $t_1 = 0$.  Then $J_{t_1} = \emptyset$.  Since $E$ is false in $\mathbf A$ at time $t_2$, it must become true due to resampling a flaw $g$ at some time $t' \geq t_2 $. Clearly $g \in \check \Gamma(E)$, i.e. $A_g e_E \not \preceq e_E$. So either $J_{t_3}$ is non-empty at time $t'$, or the rule for forming $J_{t_3 }$ at time $t'$ would add a node labeled $g$ to the empty wdag. In either case we have $J_{t_3} \neq \emptyset = J_{t_1}$.

Next, suppose $t_1 > 0$ and $J_{t_1} = J_{t_3}$. Let $f_0$ be the flaw resampled at time $0$, and consider the search algorithm $\mathbf A'$ starting at time $1$, with corresponding wdags $J'_{t_1-1}$ and $J'_{t_3-1}$ where $E$ was false at time $t_2-1$. By induction hypothesis we have $J'_{t_1 - 1} \neq J'_{t_3 - 1}$. The only way to get $J_{t_1} = J_{t_3}$ is if we prepend $f_0$ to exactly one of $J'_{t_1-1}$ or $J'_{t_3-1}$. Say it is the former (the two cases are completely symmetric). So $J_{t_1} = J_{t_3} = J'_{t_3-1}$ would have a source node labeled $f_0$. But our rule for forming wdags would then also prepend $f_0$ to $J'_{t_3-1}$ to get $J_{t_3}$, a contradiction. 

Thus, for each time $t$ that is good for $X,E$, some distinct wdag $J_t$ appears for $E$. By \Cref{jtappear}, the expected number of such times $t$ is at most $ \sum_{H \in \mathfrak J[X,E]}   \mu^{\top} A_H e_E$.
\end{proof}

\begin{corollary}
\label{distthm1d}
There holds $P(E) \leq \mu(E) \sum_{H \in \mathfrak J[ \Omega \setminus E, E]} w(H)$.
\end{corollary}
\begin{proof}
Consider the first time $t$ that $E$ becomes true, if any.  Then $E$ is false at times $0, \dots, t-1$ and so $t$ is good for $\Omega \setminus E, E$. Hence, the expected number of such times is at most $N(\Omega \setminus E, E)$. This is at most $\sum_{H \in \mathfrak J[ \Omega \setminus E, E]} \mu^{\top} A_H e_E$, which is at most $\mu(E) \sum_{H \in \mathfrak J[ \Omega \setminus E, E]} w(H)$ via \Cref{a1}.
\end{proof}

To get improved bounds of $P(E)$ via \Cref{distthm1d}, we need to analyze the wdags in $\mathfrak J[X, E]$ for the given event $E$ and for $X = \Omega \setminus E$. As a starting point, we observe two simple properties of such wdags.
\begin{observation}
\label{character-prop0}
For a wdag $H \in \mathfrak J[X,E]$, we have $\sink(H) \subseteq \check \Gamma(E)$. Furthermore, for each node $v \in H$, we have $L(v) \cap X \neq \emptyset$.
\end{observation}
\begin{proof}
For the first observation, suppose that we are building the wdag $J_t$ for a given time $t$ which is good for $X,E$, and suppose we add a sink node $v$ with label $f \notin \check \Gamma(E)$ at time $s < t$. Let $H$ denote the partial wdag before adding this node. Thus $f \not \sim H$, so $A_f A_H e_E = A_H A_f e_E$. Since $f \notin \check \Gamma(E)$, we have $A_f e_E \preceq e_E$. So $A_f A_H e_E \preceq A_H e_E$ and we would not add node $v$ to $J_t$.

For the second observation, note that if $t$ is good for $X, E$,  then by definition $X$ holds at all previous times. Since the nodes of $J_t$ are labeled by previously seen flaws, we have $L(v) \cap  X \neq \emptyset$ for any $v \in J_t$.
\end{proof}

In light of \Cref{character-prop0},  \Cref{distthm1d} gives bounds on $P(E)$ which are at least as tight as the simpler bound of \Cref{crispdist}.  For some probability spaces, additional restrictions on the structure of $\mathfrak J[X,E]$, and tighter bounds of $P(E)$, can be available. These can be quite subtle and hard to analyze. The following characterization captures some (if not all) of these restrictions.  

\begin{proposition}
\label{character-prop}
For any wdag $H \in \mathfrak J[X,E]$, there is an enumeration $\sink(H) = \{f_1, \dots, f_ k \}$ such that
\begin{equation}
\label{ord-eqn}
\forall i = 1, \dots, k \qquad A_{f_i} A_{f_{i+1}} \dots A_{f_k} e_E \not \preceq A_{f_{i+1}} \dots A_{f_k} e_E
\end{equation}
\end{proposition}
\begin{proof}
Suppose the sink nodes of $J_t$ are added at increasing times $s_1, \dots, s_k$, and let $f_i$ denote the flaw resampled at each time $s_i$. We claim that this enumeration satisfies the bound of Eq.~(\ref{ord-eqn}).  For, let $H$ denote the partial wdag $J_t$ produced immediately before prepending some $f_i$. Since $f_i$ should be the label of a sink node, we have $f_i \not \sim H$ and $\sink(H) = \{f_{i+1}, \dots, f_k \}$. We can write $A_H = A_{H'} A_{f_{i+1}} \cdots A_{f_k}$ where $H'$ are all the non-sink nodes of $H$. By our rule for forming $J_t$, we must $A_{f_i} A_H e_E \not \preceq A_H e_E$. Using the fact that $A_{H'}$ and $A_{f_i}$ commute, this implies that $ A_{H'} A_{f_i} A_{f_{i+1}} \cdots A_{f_k} e_E  \not \preceq A_{H'} A_{f_{i+1}} \cdots A_{f_k} e_E$, which further implies that $A_{f_i} A_{f_{i+1}} \cdots A_{f_k} e_E  \not \preceq  A_{f_{i+1}} \cdots A_{f_k} e_E$ as claimed.
\end{proof}

To illustrate how this characterization can give significantly stronger bounds than \Cref{crispdist} or known nonconstructive LLL estimates, we show the following result:
\begin{theorem}
\label{perm-thm}
In  the variable, permutation, or clique-perfect-matching settings, let $C$ be a stable set of atomic events, and define event $E  = \langle C \rangle$. For each $H \in \mathfrak J[X,E]$, there is an injective function $\phi_H: \sink(H) \rightarrow C$ with $f \sim \phi_H(f)$ for all $f \in \sink(H)$.
\end{theorem}

The work \cite{LLLLBeyond} showed a much stronger version of \Cref{perm-thm} for the variable setting.  The work \cite{NewBoundsHarris} showed (what is essentially) \Cref{perm-thm} for the permutation setting using a complicated and ad-hoc analysis based on a variant of witness trees. The bound for the clique-perfect-matching setting is new. We obtain all these results in a more unified way; the proofs are deferred to Appendix~\ref{perm-app}.

\begin{corollary}
\label{perm-thmaa}
In the setting of \Cref{perm-thm}, we have $$
N(X,E) \leq \mu(E) \prod_{g \in C} \bigl(1 + \sum_{f: f \sim g} \Psi_{\mathcal G}(f) \bigr) \qquad \text{ for $\mathcal G = \{f  \in \mathcal F: f \cap X \neq \emptyset \}$}.
$$ 
\end{corollary}
\begin{proof}
First, by \Cref{character-prop0},  the nodes of any wdag $H \in \mathfrak J[X,E]$ are labeled from $\mathcal G$.  Now let $C = \{g_1, \dots, g_k \}$. To enumerate such $H$, we build the set $I = \sink(H)$ by choosing, for each $i = 1, \dots, k$, a preimage set $I_{g_i} = \phi^{-1}_H(g_i)$ of cardinality at most one.  Given the choice of $I$, the remaining sum  over wdags with $\sink(H) = I$ is at most $\Psi_{\mathcal G}(I)$. Overall, this gives the bound:
$$
\sum_{I} \Psi_{\mathcal G}(I) \leq \sum_{I_{g_1}, \dots, I_{g_k} } \Psi_{\mathcal G} ( I_{g_1} \cup \dots \cup I_{g_k}) \leq  \sum_{I_{g_1}, \dots, I_{g_k} } \Psi_{\mathcal G}( I_{g_1} ) \cdots \Psi_{\mathcal G}(I_{g_k})
$$ 
where the last inequality follows from log-subadditivity of $\Psi$. This can be written as $\prod_{g \in C} \sum_{I_g} \Psi_{\mathcal G}(I_{g})$. The case of $I_{g} = \emptyset$ contributes $1$, and the case of $I_{g} = \{f \}$ contributes $\Psi_{\mathcal G}(f)$. 
\end{proof}
 
 We believe that tighter bounds on $N(X,E)$ are possible, using a more careful analysis of the structure of wdags in $\mathfrak J[X,E]$. Such bounds would lead to small improvements in the numerical constants for the later \Cref{perm-cor1}. Since the analysis needed to obtain \Cref{perm-thm} is already difficult, we leave this as an open problem.

\subsection{A distributional bound with partial dependence}
One weakness of the distributional bounds in \Cref{distrib-sec} is that the definition of $\check \Gamma(E)$ is binary: either flaw $f$ cannot possibly cause $E$, or every occurrence of $f$ must be tracked to determine if it caused $E$. The next results allow us to take account of flaws which can ``partially'' cause $E$.

For flaw $f$ and event $E$, define $$
\kappa(f, E) = \max_{\sigma \in f \setminus E} \frac{e_{\sigma}^{\top} A_f e_E}{e_{\sigma}^{\top} A_f e_{E \cup (\Omega \setminus f)}}
$$

Note that $\kappa(f, E) \in [0,1]$, and $\kappa(f,E) = 0$ for $f \notin \check \Gamma(E)$. Thus, $\kappa(f, E)$ is a weighted measure of the extent to which $f$ causes $E$.

\begin{theorem}
\label{distthm1e} 
$P(E) \leq \mu(E) +  \sum_{f \in \mathcal F}  \kappa(f, E) \cdot N (\Omega \setminus E, f \setminus E)$.
\end{theorem}
\begin{proof}
Let us say that a pair $(f,t)$ is \emph{open} if the following three conditions hold: (i) $E$ is false at times $0, \dots, t$; (ii) $f$ is resampled at time $t$; and (iii) either this is the first resampling of $f$, or if the most recent resampling of $f$ had occured at time $t' < t$, then $f$ had been false at some intermediate time between $t'$ and $t$.  We say that a triple $(f,t,s)$ with $s \geq t$ is \emph{closed} if $(f,t)$ is an open pair, and furthermore the following four conditions  all hold: (i) $f$ is resampled at time $s$; (ii) $E$ has been false at all times up to $s$; (iii) the resampling at time $s$ make $E$ true; (iv) $f$ is true at times $t, \dots, s$. Note that it is possible to have $s = t$.

We claim that if $E$ becomes true after time 0, then there is some closed triple $(f,t,s)$. For, suppose that $E$ first becomes true due to resampling flaw $f$ at time $s$. Going backward, let $t \leq s$ be the earliest time such that the same flaw $f$ is resampled at time $t$ and $f$ remained true between times $t$ and $s$. Then $(f,t)$ is open and $(f,t,s)$ is closed.

We next claim that, for a given pair $(f,t)$, the probability that there exists any closed triple $(f,t,s)$, conditional on that $(f,t)$ is open, is at most $\kappa(f,E)$. For, let $E' = E \cup (\Omega \setminus f)$, and we condition on the event that $s \geq t$ is the first time where  $f$ is resampled and $E'$ becomes true. Let $\sigma \in f \setminus E$ be the state at time $s$. The probability that $E$ holds after resampling $f$, conditional on the state $\sigma$ and that $E'$ becomes true at that resampling,  is $\frac{ e_{\sigma}^{\top} A_f e_{E}  }{e_{\sigma}^{\top} A_f e_{E'}}$, which is at most $\kappa(f,E)$ by definition.

Accordingly, the overall probability of  $E$ is at most $\mu(E) + \sum_f L_f \cdot \kappa(f, E)$, where $L_f$ is the expected number of open pairs $(f,t)$, and the first term accounts for the possibility that $E$ holds at time 0. It remains to bound $L_f$. Let us fix some flaw $f$. We claim that, for each open pair $(f,t)$, there is some distinct time $s_t$ which is good for $\Omega \setminus E,  f \setminus E$. For, going backward, find the earliest time $s_t$ such that $f$ was true from times $s_t, \dots, t$.  If $s_t > 0$, then at time $s_t - 1$ the state transits from $\Omega \setminus f$ to $f \setminus E$, which is also a transition from $\Omega \setminus (f \setminus E)$ to $f \setminus E$. Furthermore, for distinct open pairs $(f,t), (f,t')$ we have $s_t \neq s_{t'}$. So  $L_f \leq N(\Omega \setminus E, f \setminus E)$.
\end{proof}

\begin{corollary}
\label{distthm1f}
$P(E) \geq \mu(E) -  \sum_{f \in \mathcal F} \kappa(f, \Omega \setminus E) \cdot N(E, f \cap E )$.
\end{corollary}
\begin{proof}
Apply \Cref{distthm1e} to the event $\Omega \setminus E$.
\end{proof}

For certain atomically-generated resampling spaces, generic bounds on $\kappa$ are available.

\begin{proposition}
\label{kappabndgf2}
Let $\mathcal A$ be a commutative, oblivious, regenerating resampling space, which satisfies the following additional ``strong obliviousness'' property:
$$
A_f e_g = \mu(f) e_g \qquad \text{ for all atoms $f,g \in \mathcal A$ with $f \simeq g$.}
$$
(This property  holds, for example, for the variable, permutation, and clique-perfect-matching settings.)

\medskip

Then for any  event $f \in \oa$ and atom $g \in \mathcal A$ there holds
$$
\kappa(f,g) \leq \frac{\mu(g)}{1 -  \mu(f)}.
$$

 Moreover, if $f \not \sim g$, we have
$$
\kappa(f,\Omega \setminus g) \leq \frac{1 - \frac{\mu(f) \mu(g)}{\mu(f \cap g)}}{1 -  \mu(f)}
$$ 
\end{proposition}
\begin{proof}
Let $E = g$ and $E = \Omega \setminus g$ in the two cases respectively.  In either case, we can estimate the denominator in $\kappa$ as $e_{\sigma}^{\top} A_f e_{E \cup (\Omega \setminus f)} \geq  e_{\sigma}^{\top} A_f e_{\Omega \setminus f} = 1 - e_{\sigma}^{\top} A_f e_f$. We claim that, for any such event $f = \langle \{f_1, \dots, f_k \} \rangle$ and state $\sigma \in f$, there holds
\begin{equation}
\label{yyw1}
e_{\sigma}^{\top} A_f e_f \leq \mu(f)
\end{equation}

We show Eq.~(\ref{yyw1}) by induction on $k$. The case $k = 1$ holds by hypothesis.  For $k>1$, let $f' = \langle\{ f_2, \dots, f_k \} \rangle$, and let $\sigma \in f$. By \Cref{obl-p1} and some manipulations  can write:
\begin{align*}
e_{\sigma}^{\top} A_f e_f &=  \frac{\mu(f)}{\mu(f') \mu(f_1)} e_{\sigma}^{\top} A_{f_1} A_{f'}  e_{f}   \leq \frac{\mu(f)}{\mu(f') \mu(f_1)} \sum_{\sigma' \in f_1} e_{\sigma}^{\top} A_{f_1} e_{\sigma'} \cdot  e_{\sigma'}^{\top} A_{f'} e_{f'}  \\
& \leq \frac{\mu(f)}{\mu(f') \mu(f_1)} \sum_{\sigma' \in f_1} e_{\sigma}^{\top} A_{f_1} e_{\sigma'} \mu(f') =\frac{\mu(f)}{\mu(f_1)}   e_{\sigma}^{\top} A_{f_1}e_{f_1}.  & \text{(induction hypothesis)}
\end{align*}
By hypothesis, is equal to $\mu(f)$.

Thus, the denominator in $\kappa$ (in either case) is at least $1 - \mu(f)$. We turn to estimating the numerator. For $\kappa(f,g)$, we claim that $e_{\sigma} A_f e_g = \mu(g)$ for any state $\sigma \in f$. For, consider $f = \langle \{ f_1, \dots, f_k \} \rangle$; by \Cref{obl-p2}, any reordering of $f_1, \dots, f_k$ would give the same matrix $A_f$.  So assume without loss of generality that $f_k \sim g$. When implementing $\mathbf R_f$, suppose we condition on the state $\sigma'$ just before resampling $f_k$. By hypothesis, the probability that $\sigma'$ gets mapped to $g$ is precisely $\mu(g)$.

For $\kappa(f, \Omega \setminus g)$ with $g \not \sim f$, \Cref{obl-p00} gives \[
e_{\sigma}^{\top} A_f e_E = 1 - e_{\sigma}^{\top} A_f e_g = 1 - e_{\sigma}^{\top} \frac{\mu(f) \mu(g)}{\mu(f \cap g)} e_{f \cap g} = 1 -\frac{\mu(f) \mu(g)}{\mu(f \cap g)}. \qedhere \]
\end{proof}

\subsection{Applications}  Using the more sophisticated distributional bounds we can show the following bounds for the permutation and perfect-matching probability spaces.
\begin{theorem}
\label{perm-cor1}
\begin{enumerate}
\item If each color appears at most $\Delta = \frac{27}{256} n$ times in the array, then the Search Algorithm generates a latin transversal where, for each cell $x,y$, there holds
$$
\frac{17}{32 n} \leq P( \pi x = y ) \leq \frac{203}{128 n}
$$
\item Consider an edge-coloring $C$ of the clique $K_n$, for $n$ even, such that each color appears on at most $\Delta = \frac{27}{256} n$ edges. Then the Search Algorithm generates a perfect matching $M$ such that $C(e) \neq C(e')$  for all distinct edges $e, e'$ of $M$.  Moreover each edge $e$ has
$$
\frac{17}{32 (n-1)} \leq P( e \in M) \leq \frac{203}{128(n-1)}.
$$
\end{enumerate}
\end{theorem}
\begin{proof}
  We show only the result on permutations; the result for the clique is completely analogous.

Let $g$ be the atomic event $g$ that $\pi x = y$.   Note that any time that is good for $\Omega \setminus g, f \setminus g$ is good for $\Omega, f$ as well, so $N( \Omega \setminus g, f \setminus g) \leq N(\Omega,  f)$. Thus, for the upper bound, \Cref{distthm1e} gives:
  $$
  P( \pi x = y) \leq \mu(g) + \sum_{f \in \mathcal F} N(\Omega, f) \kappa(f, g)
  $$
  
  Now consider a flaw $f \sim g$ defined by atoms $f_1, f_2$. By \Cref{kappabndgf2}, we have
  $$
  \kappa(f,g) \leq \frac{1/n}{1 - \frac{1}{n(n-1)}}
   $$
  
  By  \Cref{perm-thmaa}, we have $N( \Omega, f) \leq \frac{ (n-2)! }{n!} \prod_{i=1}^2 \bigl(1 + \sum_{f' \in \mathcal F: f' \sim f_i} \Psi(f') \bigr)$. Since $\Psi(f') \leq \gamma := \frac{256}{81 n^2}$, and for each $i = 1,2$ there are at most $2 n (\Delta-1)$ choices for $f'$, we overall get $ N(\Omega, f) \leq \frac{ ( 1 + 2 n (\Delta-1) \gamma)^2}{n(n-1)}$. There are at most $2 n (\Delta - 1)$ choices for $f$, so we have
 $$
  P(g) \leq \frac{1}{n}  + 2 n (\Delta-1) \cdot \frac{( 1 + 2 n (\Delta-1) \gamma \bigr)^2}{n(n-1)} \cdot \frac{ 1/n}{1 - \frac{1}{n(n-1)}}
  $$
  and routine calculations show that this is at most $203/(128 n)$.

For the lower bound, we use \Cref{distthm1f} to get:
$$
P(g) \geq \mu(g) - \sum_{f \in \mathcal F} N( g, f \cap g)  \kappa(f, \Omega \setminus g)
  $$
  
We will bound $N( g, f \cap g)$ via \Cref{perm-thm}. Let  $\mathcal G = \{f \in \mathcal F: f \cap g \neq \emptyset \}$. Consider some such flaw $f$ corresponding to atoms $f_1, f_2$. First, suppose $f_1, f_2, g$ are distinct.  Then   \Cref{perm-thmaa} gives $$
N(  g, f \cap g) \leq \frac{1}{n(n-1)(n-2)}  \bigl( 1 + \sum_{f': f' \sim g} \Psi_{\mathcal G}(f') \bigr) \prod_{i=1}^2 \bigl(1 + \sum_{f': f' \sim f_i} \Psi_{\mathcal G}(f') \bigr)
$$

We can estimate $ \sum_{f': f' \sim f_i} \Psi_{\mathcal G}(f') \leq \sum_{f': f' \sim f_i} \Psi_{\mathcal F}(f')  \leq 2 n (\Delta - 1) \gamma$ for each $i = 1,2$. Also, $\mathcal G$ does not contain any flaws $f$ with $f \sim g$, so the term $\sum_{f': f' \sim g} \Psi_{\mathcal G}(f')$ contributes nothing. Overall we have $N(  g, f \cap g) \leq \frac{ (1 + 2 n (\Delta - 1) \gamma)^2}{n(n-1)(n-2)}$.  There are at most $\frac{n^2 (\Delta-1)}{2}$ flaws of this kind. By \Cref{kappabndgf2}, each such flaw $f$ has $\kappa(f, \Omega \setminus g) \leq \frac{1 - \frac{(1/n(n-1)) (1/n)}{1/n(n-1)(n-2)}}{1 - 1/n(n-1)} = \frac{2/n}{1-1/n(n-1)}$.

Next, suppose that say $f_1 = g$. Then, by the same reasoning as above, \Cref{perm-thmaa}  gives $N( g, f \cap g) = N( g, f) \leq \frac{(1 + 2 n (\Delta - 1) \gamma)}{n(n-1)}$. There are at most $(\Delta-1)$ flaws of this kind, and each trivially has $\kappa(f, \Omega \setminus g) \leq 1$.   

    Putting all terms together, we have
 \begin{align*}
  &P(g)  \geq 1/n  -  n^2 (\Delta-1)/2  \cdot \frac{ (1 + 2 n (\Delta - 1) \gamma)^2}{n(n-1)(n-2)} \cdot \frac{ 2/n }{1 - \frac{1}{n(n-1)} } -  (\Delta-1) \frac{ (1 + 2 n (\Delta - 1) \gamma)}{n(n-1)} \cdot 1 \thickspace
  \end{align*}
  which is easily seen to be at least $\frac{17}{32 n}$. 
  \end{proof}

  Note that \Cref{crispdist} would yield a weaker bound $P( \pi x = y) \leq \frac{16}{9 n}$, and \Cref{perm-thm} directly would yield a weaker bound $P( \pi x = y) \leq \frac{5}{3 n}$. It is not known if a constant term better than $16/9$ can be shown for the LLL-distribution.

\appendix

\section{Proof of \Cref{perm-thm}}
\label{perm-app}
For brevity throughout, for a stable set $I$ of $\mathcal A$, we write $e_I$ as shorthand for $e_{\langle I \rangle}$.

 As an easy warm-up exercise, we consider the variable setting. To recall and set notation, $\Omega$ is the cartesian product distribution on tuples $X = (X_1, \dots, X_n)$. For each index $i = 1, \dots, n$ and value $y$, there is atomic event $X_i = y$; for brevity in this section, we denote this atom by $[i, y]$. Its resampling oracle draws $X_i$ again from its original distribution. We have $[i,y] \sim [i', y']$ if $i = i'$ and $y \neq y'$.

\begin{proposition}
\label{add-1-prop_00}
For any $\mathcal A$-stable multiset $I$, and corresponding set $\bar I$ (i.e. where we keep at most one copy of each element in $I$), there holds $A_I e_C \propto e_{D}$ for stable set $D = \bar I \cup (C \setminus \Gamma(\bar I))$.
\end{proposition}
\begin{proof}
We show this by induction on $I$. The base case $|I| = 0$ is trivial since then $D = C$. For the induction step, consider some atom $f = [i,y]$, and let $I' = I \uplus \{ f \}$. By induction hypothesis, we have $A_{I'} = A_f A_I e_C \propto A_f e_D$ for $D = \bar I \cup (C \setminus \Gamma(\bar I))$. So it suffices to show that $A_f e_D \propto e_{D'}$ for $D' = \bar I' \cup (C \setminus \Gamma(\bar I')) = D \cup \{f \} \setminus \Gamma(f)$.

Consider a state $X$. If $X \notin f$, then $e_{X}^{\top} A_f e_D = 0 = e_{X}^{\top} e_{D'}$. Similarly, if $X\notin f'$ for $f' \in D \setminus \Gamma(f)$, then also $e_{X}^{\top} A_f e_D = 0 = e_{X}^{\top} e_{D'}$ since resampling $f$ cannot move the state to $f'$.  Thus, we suppose that $X \in \langle D' \rangle$.  Since $C$ is a stable set, it has at most one neighbor of $f$. If $C \cap \Gamma(f)  = \emptyset$, then $e_X^{\top} A_f e_C = 1$ for any such state $X$; otherwise, if $C \cap \Gamma(f) = \{ [i, y'] \}$, then $e_X^{\top} A_f e_C = \Pr_{\Omega}(X_i = y')$ for any such state $X$. In either case, it is a scalar which is constant for all $X \in \langle D' \rangle$.
\end{proof}

\begin{proof}[Proof of \Cref{perm-thm} for the variable setting]
Enumerate $I = \{ f_1, \dots, f_k \}$ to satisfy \Cref{character-prop}, where $f_i = \langle F_i \rangle$. For each $i$, let $J_i = F_1 \uplus \dots \uplus F_i$ and let $\bar J_i = F_1 \cup \dots \cup F_k$. We claim that, for each $i = 1, \dots, k$, we have
\begin{equation}
\label{zaz1eq}
\bar J_i \cap \Gamma(C) \neq \bar J_{i-1} \cap \Gamma(C)
\end{equation}

For, suppose that $\bar J_i \cap \Gamma(C) = \bar J_{i-1} \cap \Gamma(C)$, and define $I' = \{f_1, \dots, f_{i-1} \}$. By \Cref{obl-p2} we have $A_{I'} \propto A_{J_{i-1}}$ and by \Cref{add-active-prop}, we have $A_{J_{i-1}} e_E \propto e_{D'}$ where $D' =\bar J_{i-1} \cup (C \setminus \Gamma(\bar J_{i-1})$. Thus, we can write $A_{I'} e_E = p e_{D'}$ for some scalar  $p \geq 0$.  Now consider any state $X$; we claim that 
\begin{equation}
\label{zbz1eq}
e_{X}^{\top} A_{f_i} A_{I'} \leq p \cdot e_{X}^{\top} e_{D'}
\end{equation}
 Let $D = \bar J_i \cup (C \setminus \Gamma(\bar J_i))$; since $\bar J_i \cap \Gamma(C) = \bar J_{i-1} \cap \Gamma(C)$ we have $D \supseteq D'$. So the LHS of (\ref{zbz1eq}) is zero unless $X \in \langle D \rangle \subseteq \langle D' \rangle$. On the other hand, for $X \in \langle D \rangle$, the substochasticity of matrix $A_{f_i}$ gives $  e_{X}^{\top} A_{f_i} A_{I'} e_E = e_{X}^{\top} A_{f_i} \cdot p e_{D'}  \leq p$  as desired.
  
  So, if $\bar J_i \cap \Gamma(C) = \bar J_{i-1} \cap \Gamma(C)$, we would have $  A_{f_i} A_{I'} \preceq A_{I'}$, contradicting \Cref{character-prop}.  Thus, contrariwise, we have established Eq.~(\ref{zaz1eq}) for $i =1, \dots, k$. We define the function $\phi$ by setting $\phi( f_i ) = g_i$ where $g_i$ is an arbitrary element in $(\bar J_i \cap \Gamma(C)) \setminus (\bar J_{i-1} \cap \Gamma(C))$.
\end{proof}

\bigskip
We now turn to the much harder  permutation setting. To recall and set notation, $\Omega$ is the uniform distribution on permutations $\pi$ on $\{1, \dots, n \}$. For each pair $(x,y)$, there is atomic event $\pi x = y$; for brevity in this section, we denote this event by $[x, y]$. Its resampling oracle sets $\pi \leftarrow (y \ z) \pi$, for $z$ drawn uniformly from $\{1, \dots, n\}$. We have $[x,y] \sim [x', y']$ if exactly one of the following holds: (i) $x = x'$ or (ii) $y = y'$. 

\begin{proposition}
\label{add-1-prop}
Let $f = [x,y]$ be an atom of $\mathcal A$ and let $C$ be a stable set of $\mathcal A$. Then $A_f e_C \propto e_{C'}$ for the $\mathcal A$-stable set $C'$ obtained from $C$ as follows:
\begin{itemize}
\item If $C$ contains two neighbors $f_1 = [x,y_1], f_2 = [x_2,y]$ of $f$, then $C' = C \cup \{ [x_2, y_1], f \} \setminus \Gamma(f)$.
\item Otherwise,  $C' = C \cup \{ f \} \setminus \Gamma(f)$.
\end{itemize}
\end{proposition}
\begin{proof}
We want to show that there is a scalar $p$ with $e_{\pi}^{\top} A_f e_C = p \cdot e_{\pi} e_{C'}$ for all states $\pi$.  If $\pi \notin f$, then $e_{\pi}^{\top} A_f e_C = 0 = e_{\pi}^{\top} e_{C'}$. Similarly, if $\pi \notin f'$ for $f' \in C \setminus \Gamma(f)$, then also $e_{\pi}^{\top} A_f e_C = 0 = e_{\pi}^{\top} e_{C'}$ since resampling $f$ cannot move the state to $f'$.  Thus, we suppose that $\pi \in \langle C \cup \{f \} \setminus \Gamma(f) \rangle$. 

Now suppose we resample $f$ to $\pi' = (y \ z) \pi$. We consider the cases in turn:
\begin{itemize}
\item If $C$ contains two neighbors $f_1 = [x,y_1], f_2 = [x_2,y]$, then we claim that $\pi' \in \langle C \rangle$ precisely when $z = \pi x_2 = y_1$. For, in order to get $\pi' \in f_1$,  we must have $\pi' x = y_1$. Since $\pi x = y$, this implies $(y \ z) y = y_1$, i.e. $z = y_1$. Thus, $\pi' = (y \ y_1) \pi$. To get $\pi' \in f_2$, we need $y = \pi' x_2 = (y \ y_1) \pi x_2$, i.e. $\pi x_2 = y_1$.   In particular, we must have $\pi \in [x_2, y_1]$, and for any such $\pi$ we have $e_{\pi}^{\top} A_f e_C = 1/n$ and $e_{\pi}^{\top} e_{C'} = 1$.

\item If  $C$ has a single neighbor $f_1 = [x, y_1]$ of $f$, then $\pi' \in f_1$ precisely if $y_1 = z$. Similarly, if $C$ has a single neighbor $f_2 = [x_2, y]$ of $f$, then $\pi' \in f_2$ precisely if $z = \pi x_2$.  In either case, we have $e_{\pi}^{\top} A_f e_C = 1/n$ for all such $\pi$ and also $e_{\pi}^{\top} e_{C'} = 1$.

\item If $f \in C$, then $\pi'$ is in $\langle C \rangle$ iff $z =  y$ and $\pi' = \pi \in \langle C \rangle$. Then $e_{\pi}^{\top} A_f e_C = \frac{1}{n} e_C$. Note that, because $C$ is a stable set, this case implies that $C$ has no neighbors of $f$.

 \item If $C \cap \overline \Gamma(f) = \emptyset$, then $\pi'$ is in $\langle C \rangle$ iff $z \notin \{ y_1, \dots, y_k \}$ where $C = \{ [x_1, y_1], \dots, [x_k, y_k] \}$. Thus $e_{\pi}^{\top} A_f e_C = \frac{n - k}{n}$ and $e_{\pi}^{\top} e_{C'} = 1$. \qedhere
 \end{itemize}
\end{proof}

\begin{proposition}
\label{add-active-prop}
Given stable multisets $C, I$ of $\mathcal A$, let $\bar I$ denote the corresponding set of $I$ (i.e. with at most one copy of each element).  Define a bipartite graph $G^C_{I}$ with left-vertex-set $C$ and right-vertex-set  $\bar I$, with an edge on $f, f'$ iff $f \sim f'$, and let $\tau^{C}(\bar I)$ be the size of a maximum matching in $G^C_{I}$.   

For such multisets $I, C$, there is an associated stable set $D^C_{I}$ of $\mathcal A$ with  $A_I e_C \propto e_{D^C_{I}}$. Furthermore, if $\tau^C(\bar I \cup \{f \}) = \tau^C(\bar I)$, then $D^C_{I \uplus \{f \}} = D^C_{I} \cup \{ f \}$.
\end{proposition}
\begin{proof}
Throughout this proof, we fix $C$ and write $D_I, \tau(\bar I), G_I$ instead of $D^C_{I}, \tau^C(\bar I), G^C_{I}$ etc.

Since $C$ and $\bar I$ are stable sets, the graph $G_{I}$ has degree at most two ---  each node $[x,y]$ has at most one neighbor of the form $[x', y]$ and at most one neighbor of the form $[x, y']$. So $G_{I}$ decomposes into paths and cycles and any maximal path of $G_I$ which starts and ends at left-nodes (which we call a \emph{$C$-path}), can be written uniquely as
$$
[x_1, y_1], [x_1, y_2], [x_2, y_2], \dots, [x_k, y_{k-1}], [x_k, y_k];
$$
We define $H_I$ to be the set of atoms $[x_k, y_1]$ for each such $C$-path; this includes the case $k = 1$ where $[x_1, y_1]$ is an isolated $C$-node. We define $D_I = \bar I \cup H_I$; note that $D_I$ is an $\mathcal A$-stable set.

We first show that $A_I e_C \propto e_{D_{I}}$ by  induction on $|I|$. The base case $I =\emptyset$ holds since then $D_{I} = C$. For the induction step, let $I' = I \uplus \{f \}$. By induction hypothesis, we have $A_{I'} e_C = A_f A_{I} e_C \propto A_f e_{U}$ for $U = D_{I}$. This in turn is proportional to $e_{ U' }$ for the stable set $U'$ obtained from $U, f$ according to the rules given in \Cref{add-1-prop}. There are three cases.

 If $U$ has no neighbors of $f$, or if $f \in D_I$, then $U' = U \cup \{f \}$,  and $H_{I'} = H_{I} \cup \{f \}$. So indeed $U' =  D_{I'} = D_I \cup \{ f \}$. Thus, we assume for the remainder that $f \notin I$.

 If $U$ has one neighbor $g = [x_1, y_1]$ of $f$, then since $\bar I$ is a stable set, it must have $g \notin \bar I$, i.e. $G_{I}$ contains a $C$-path $P$ with endpoints $x_1, y_1$. In $G_{I'}$, this $C$-path now terminates in a degree-one left-node $[x, y]$, and $P$ is no longer a $C$-path of $G_{I'}$. So $D_{I} \cup \{f \} \setminus \{ g \} = U' = D_{I'}$.

 If $U$ has two neighbors $g_1 = [x, y_1], g_2 = [x_2, y]$, then again since $\bar I$ is a stable set these must correspond to $C$-paths in $G_{I}$. Thus, $G_{I}$ has two $C$-paths with endpoints $x, y_1$ and $x_2, y$ respectively. Now $G_{I'}$ has a new left-node $[x,y]$. This merges the two $C$-paths into a single new $C$-path with endpoints $x_2, y_1$. Thus again $D_{I'} = D_{I} \cup \{ f, [x_2, y_1] \} \setminus \{g_1, g_2 \}  = U'$.
 
 This completes the induction. Next, suppose $\tau^C(\bar I \cup \{f \}) = \tau^C(\bar I)$. If $f \in C$ or $f \in I$, then  $f \in D_{I}$ and we have already seen that $D_{I'} = D_{I} = D_I \cup \{f \}$. So suppose $f$ is not in $C$ or $I$,  and let $\tilde G'$ denote the connected component  of $G_{I'}$ containing  $f$, and let $\tilde G$ be the  graph obtained by deleting $f$ from $\tilde G'$.  Since $\tilde G'$ is a path or cycle, and $\tilde G$ is obtained by deleting a right-node, the only way that $\tilde G$ and $\tilde G'$ can have the same matching size is if  $\tilde G'$ is a path starting and ending at right-nodes. In this case, neither $\tilde G'$ nor $\tilde G'$ have any $C$-paths. Hence we get $H_{I \uplus \{f \}}  = H_{I}$. 
\end{proof}

\begin{proof}[Proof of \Cref{perm-thm} for the permutation setting]
Enumerate $I = \{ f_1, \dots, f_k \}$ to satisfy \Cref{character-prop}, where $f_i = \langle F_i \rangle$. For each $i$, let $J_i = F_1 \uplus \dots \uplus F_i$ and let $\bar J_i = F_1 \cup \dots \cup F_k$. We claim that, for each $i = 1, \dots, k$, we have
\begin{equation}
\label{zaz2eq}
\tau^C( \bar J_i) > \tau^C( \bar J_{i-1})
\end{equation}

For, suppose that $\tau^C(\bar J_i) = \tau^C(\bar J_{i-1})$, and define $I' = \{f_1, \dots, f_{i-1} \}$. We have $A_{I'} \propto A_{J_{i-1}}$ and by \Cref{add-active-prop}, we have $A_{J_{i-1}} e_E \propto e_{D'}$ for some stable set $D'$ of $\mathcal A$. Thus, we can write $A_{I'} e_E = p e_{D'}$ for some scalar  $p \geq 0$.  Now consider any state $\pi$; we claim that 
\begin{equation}
\label{zbz2eq}
e_{\pi}^{\top} A_{f_i} A_{I'} \leq p \cdot e_{\pi}^{\top} e_{D'}
\end{equation}

Again, by \Cref{add-active-prop}, we have $A_{f_i} A_{I'} \propto e_D$ for a stable set $D$; since $\tau^C(F_i \cup \bar J_{i-1}) = \tau^C(\bar J_{i-1})$ we have $D = D' \cup F_i$.   So the LHS of (\ref{zbz2eq}) is zero unless $\pi \in \langle D \rangle \subseteq \langle D' \rangle$. On the other hand, for $\pi \in \langle D \rangle$, the substochasticity of matrix $A_{f_i}$ gives $  e_{\pi}^{\top} A_{f_i} A_{I'} e_E = e_{\pi}^{\top} A_{f_i} \cdot p e_{D'}  \leq p$  as desired.
  
  So, if $\tau^C(\bar J_i) = \tau^C(\bar J_{i-1})$, we would have $  A_{f_i} A_{I'} \preceq A_{I'}$, contradicting \Cref{character-prop}.  Thus, contrariwise, we have established  Eq.~(\ref{zaz2eq}) for $i =1, \dots, k$. So,  for each $i$ there is $g_i \in F_i$ and $F_i' \subseteq F_i \setminus \{g_i \}$ with $\tau^C(\bar J_{i-1} \cup F_i' \cup \{g_i \}) > \tau^C(\bar J_{i-1} \cup F'_i)$. 

It is known (see, e.g. \cite[Example 1.4]{lovasz1983}) that $\tau^C$ is a submodular set function for fixed $C$. Hence, 
$$
1 = \tau^C(\bar J_{i-1} \cup F_i' \cup \{g_i \}) - \tau^C(\bar J_{i-1} \cup F'_i) \leq \tau^C(\{g_1, \dots, g_{i-1} \} \cup \{g_i \}) - \tau^C( \{g_1, \dots, g_{i-1} \})
$$
since $\{g_1, \dots, g_{i-1} \} \subseteq \bar J_{i-1}$.  So $\tau^C( \{g_1, \dots, g_k \}) = k$ and $G^C_{ \{ g_1, \dots, g_k \} }$ has a matching $M$ of size $k$.  We define the function $\phi$ by setting $\phi( f_i ) = c_i$ where $g_i$ is matched to $c_i$ in $M$.
\end{proof}

 The clique-perfect-matching setting is very similar to the permutation setting. Here, $\Omega$ is the uniform distribution on perfect matchings $M$ of the $n$-clique. For each pair $(x,y)$ with $x \neq y$, there is an atom that $M \supseteq \{x, y \}$; we denote this by $[x,y]$. Note that $[x,y] = [y,x]$. For $x < y$, the resampling oracle is to draw $z$ uniformly from $\{1, \dots, n \} \setminus \{x \}$ and set $M \leftarrow (y \ z) M$. We define $[x,y] \sim [x', y']$ iff $| \{x, y \} \cap \{x' , y' \} | = 1$.

\begin{proposition}
\label{add-1-prop2}
Let $f = [x,y]$ be an atom of $\mathcal A$ where $x < y$ and let $C$ be a stable set of $\mathcal A$. Then $A_f e_{C} \propto e_{C'}$ for the stable set $C'$ obtained from $C$ as follows:
\begin{itemize}
\item If $C$ contains exactly two neighbors $f_1 = [x,y_1], f_2 = [x_2,y]$ of $f$, then $C' = C \cup \{ f, [x_2, y_1] \} \setminus \Gamma(f)$.
\item Otherwise, $C' = C \cup \{f \} \setminus  \Gamma(f)$.
\end{itemize}
\end{proposition}

\begin{proposition}
\label{add-active-prop22}
Given stable multisets $C, I$ of $\mathcal A$, let $\bar I$ denote the corresponding set of $I$. Define a bipartite graph $G^C_{I}$ with left-vertex-set $C$ and right-vertex-set  $\bar I$, with an edge on $f, f'$ iff $f \sim f'$, and let $\tau^{C}(\bar I)$ be the size of a maximum matching in $G^C_{I}$.   

For such multisets $I, C$, there is an associated stable set $D^C_{I}$ of $\mathcal A$ with $A_I e_C \propto e_{D^C_{I}}$. Furthermore, if $\tau^C(\bar I \cup \{f \}) = \tau^C(\bar I)$, then $D^C_{I \uplus \{f \}} = D^C_{I} \cup \{ f \}$.
\end{proposition}

We omit the proofs of \Cref{add-1-prop2} and \Cref{add-active-prop22}, as well as the remainder of the proof of \Cref{perm-thm}, as they are precisely analogous to the permutation setting.

\bibliographystyle{plain}
\bibliography{kolmo}

\end{document}